\documentclass[10pt]{article}
\usepackage{bm,amsfonts,amssymb,amsthm,amsmath}
\usepackage{mathrsfs}
\usepackage{graphicx,color}
\usepackage{CJK}
\usepackage{indentfirst,enumerate}
\usepackage{multirow}

\DeclareMathAlphabet{\mathpzc}{OT1}{pzc}{m}{it}

\newcommand{\md}{\mathrm{d}}
\newcommand{\tsn}[2]{{#2}^{#1}}
\newcommand{\idn}{\mathfrak{i}}

\newtheorem{theorem}{Theorem}
\numberwithin{theorem}{section}
\newtheorem{lemma}[theorem]{Lemma}
\newtheorem{proposition}[theorem]{Proposition}
\newtheorem{corollary}[theorem]{Corollary}

\numberwithin{equation}{section}

\definecolor{orange}{rgb}{1,0.5,0}
\definecolor{rb}{rgb}{1,0,1}
\allowdisplaybreaks[4]

\begin{document}

\title{\textbf{Symmetry-consistent expansion of interaction kernels between rigid molecules}\footnote{This work is partially supported by NSFC No. 11688101, NCMIS, AMSS grants for Outstanding Youth Project, and ICMSEC dirctor funds.}}
\author{Jie Xu\footnote{LSEC \& NCMIS, Institute of Computational Mathematics and Scientific/Engineering Computing (ICMSEC), Academy of Mathematics and Systems Science (AMSS), Chinese Academy of Sciences, Beijing 100190, China. Email: xujie@lsec.cc.ac.cn}}
\date{}
\maketitle
\begin{abstract}
  We discuss the expansion of interaction kernels between anisotropic rigid molecules. The expansion decouples the correlated orientational variables so that it can be utilized to derive macroscopic models. 
  Symmetries of two types are considered.
  First, we examine the symmetry of the interacting cluster, including the translation and rotation of the whole cluster, and label permutation within the cluster.
  The expansion is expressed by symmetric traceless tensors, and the linearly independent terms are identified. 
  Then, we study the molecular symmetry characterized by a point group in $O(3)$. 
  The proper rotations determine what symmetric traceless tensors can appear. 
  The improper rotations decompose these tensors into two subspaces and determine how the tensors in the two subspaces are coupled. 
  For each point group, we identify the two subspaces, so that the expansion consistent with the point group is established. 
\end{abstract}





\section{Introduction}


In a system consisting of many rigid molecules, the interactions between the molecules depend not only on the relative position, but also on the relative orientation. 
Such interactions can lead to nonuniform orientational distribution. 
As a result, even in an infinitesimal volume, local anisotropy can be formed and further correlated spatially, which is the typical mechanism for liquid crystals. 
An example that many people are familiar with is the (uniaxial) nematic phase formed by rod-like molecules, where no positional order is observed but an optical axis can be identified. 
If layer structure further arises, the smectic phases could appear. 
The concept of liquid crystals has been expanded to a great extent 
since rigid molecules of other shapes, such as bent-core molecules, have proved to possess richer phase behavior experimentally \cite{JJAP,jakli2018physics}. 

In mathematical theory, to identify liquid crystalline phases, one needs to construct free energy about some order parameters describing the local anisotropy.
A simple approach is to construct phenomenological models, typically a polynomial of the order parameters and their derivatives. 
For rod-like molecules, the order parameter can be chosen as a second order symmetric traceless tensor, based on which the Landau-de Gennes theory is built 
and has been successfully applied to both stationary and dynamic problems \cite{GennesPierreGillesde1993Tpol,Beris_book,qian1998generalized}. 
When discussing other types of liquid crystalline phases, including polar, biaxial or tetrahedral order, people also attempted to construct phenomenological models with different tensor order parameters \cite{fel1995tetrahedral,0295-5075-54-2-206,de2008landau,trojanowski2012tetrahedratic,shamid2014predicting,gaeta2016octupolar}. 

Despite the success of phenomenological theories, they still do not touch an essential problem: to understand the connection between the molecular architecture and macroscopic phenomena. 
%
To accompolish this goal, it is desirable to build macroscopic theory from molecular interactions.
Molecular interactions are characterized by kernel functions of several molecules, in which the variables representing the positions of these molecules are correlated. 
To derive a macroscopic theory, it is necessary to decouple these variables, which is attained by expanding the kernel functions. 
Such an approach dates back to the derivation of the equations of state for gases, where a homogeneous system consisting of spherical molecules is considered \cite{Mayer_book,Landau_Lifshitz_Statistical}. 
Inhomogeneous systems, without considering the anisotropy of the molecule, have also been discussed, leading to theories for modulated phases that can be used to describe various materials such as amphiphilic systems and block copolymers \cite{LB,ohta1986equilibrium,fredrickson1987fluctuation}.

When non-spherical rigid molecules are put into considertaion, extra variables are introduced for the orientation of the molecule. 
Most theories developed from molecular interactions focus on the orientational variables only and are built for spatially homogeneous systems.
In this case, the kernel functions are independent of spatial variables, and the expansion decouples the orientational variables. 
Theories of this kind possibly start from 
Maier--Saupe \cite{M_S} for rod-like molecules. 
Other rigid molecules, including cuboid, bent-core, triangle and cross-like \cite{Bi1,PhysRevE.58.5873,doi:10.1063/1.1649733,SymmO,xu2017transmission}, have also been discussed. 

Recently, the expansion has been extended to spatially inhomogeneous cases, so that both spatial inhomogeneity and orientational anisotropy are included.
This approach combines the techniques for spatially inhomogeneous systems of spherical molecules and for spatially homogeneous systems of non-spherical molecules. 
It was first proposed for rod-like molecules \cite{RodModel}, for which a tensor model was established for both nematic and smectic phases. 
Later, it has been successfully applied to bent-core molecules \cite{BentModel}, resulting in a tensor model for modulated nematic phases. 

Symmetry is a central topic when discussing the expansion of interaction kernels between non-spherical rigid molecules, which can be expressed by a few arguments for interaction kernels. 
Some of the arguments are spontaneous, while the others originate from molecular symmetry. 
These arguments determine what terms will appear in the expansion, 
and each term in the expansion leads to a term expressed by tensors in the free energy. 
Thus, together with suitable truncation, the molecular symmetry determines the form of the free energy, as well as the order parameters that are just the tensors appearing in the free energy. 
%
%
When the interaction kernels are specified, the coefficients can be calculated from the kernels \cite{Bi1,PhysRevE.58.5873,SymmO,RodModel,BentModel}. 
Molecules with the same symmetry, 
such as bent-core and star-shaped molecules \cite{BentModel,doi:10.1080/02678292.2017.1290285,xu2018onsager-theory-based}, can be distinguished in this way. 


The works mentioned above indicate that the expansion of interaction kernels is the core of deriving macroscopic models from molecular interactions. 
These works, however, only discuss particular molecular symmetries, not covering other rigid molecules exhibiting interesting phenomena \cite{JJAP,kim2011tetrahedratic,Zhao29092015}. 
In this work, we derive the expansion of interaction kernels for \emph{all} molecular symmetries. 
The task is carried out in two steps. 
\begin{itemize}
\item Write down the general form of expansion without considering molecular symmetry. In this case, the interaction kernels still possess some spantaneous symmetry arguments. We shall figure out the role of these arguments playing on the expansion. 
\item With certain molecular symmetry, some terms in the general expansion will vanish. We shall identify the nonvanishing terms for all molecular symmetries. 
\end{itemize}
Moreover, we will discuss the interaction kernels involving clusters of multiple molecules, with explicit expressions written down for clusters of up to four molecules. 
In most previous works, only pairwise (two-molecule clusters) interaction is taken into account. 


To deal with anisotropy at molecular level, we need to introduce the orientational variables. 
For this purpose, some notations and results about $SO(3)$ are presented in Section \ref{notations}. 
Then, in Section \ref{expansion_general}, we study the expansion of interaction kernels in the general case. 
Whatever the molecular potential is, the interaction kernel for a cluster is invariant when the whole cluster is displaced and rotated as a whole, and the labels of molecules within the cluster are permuted. 
To reveal the effect of these arguments, one crucial point is that we express the expansion by symmetric traceless tensors that have been discussed in \cite{xu_tensors}. 
With symmetric traceless tensors, it is easier to identify the linearly independent terms. 
The orthogonality of many terms can also be recognized, so that approximation results can be established. 

The use of symmetric traceless tensors also makes it clear how the molecular symmetry plays its role, 
which we analyze in Section \ref{expansion_sym}. 
The molecular symmetry is described by orthogonal transformations leaving the molecule invariant, which form a point group in $O(3)$.
A point group consists of proper rotations and possibly improper rotations, whose roles are different. 
The proper rotations constitute a subgroup in $SO(3)$. 
This subgroup determines that only the invariant tensors of this group can appear in the expansion. 
The invariant tensors have been written down explicitly in \cite{xu_tensors} for each point group in $SO(3)$. 
Then, the improper rotations decompose the invariant tensors into two orthogonal subspaces, and impose conditions on the coupling of tensors between these two subspaces. 
For each point group, we will write down explicitly the subspace decomposition. 
If the two point groups have the same proper rotations, the invariant tensors are identical. 
However, since the decomposition by improper rotations is distinct, the surviving terms in the free energy would be different. 

In this way, we write down the expansion of the interaction kernel for each point group. 
For the expansion in the general case, the list of all the linearly independent terms is provided; the proper rotations select tensors that appear in these terms; the improper rotations set the rule of coupling. 
This procedure clearly reflects how the molecular symmetry selects terms in the expansion, 
which is summarized in Section \ref{summary}. 

The current work can serve as a useful handbook for studying the liquid crystalline phases formed by any rigid molecule. 
When studying a particular rigid molecule, one could look up the list and choose the terms needed, no matter for molecular-based theories or for Landau theories. 
If one would like to evaluate the coefficients from the kernel, the orthogonality will also help. 
%
To actually write down a free energy, one needs to truncate at certain order. 
The truncation criteria might be influenced by stability, symmetry of the phase, degrees of freedom of macroscopic parameters \cite{mettout2006macroscopic,BentModel,xu_tensors}. 
From the complete list of terms we provide, 
one could choose the terms based on the need for particular systems. 






\section{Preliminaries\label{notations}}

We consider the system consisting of many identical rigid molecules that are generally anisotropic, so that the orientation of each molecule affects the state of the whole system. 
To describe the orientation, we mount a right-handed orthonormal frame $(\hat{O}; \bm{m}_1,\bm{m}_2,\bm{m}_3)$ on the molecule. 
The position of $\hat{O}$ is denoted by $\bm{x}$, and the orientation of the frame is denoted by $\mathfrak{p}$.
In this way, $(\bm{x},\mathfrak{p})$ represents the position and the orientation of the molecule. 
The frame $\mathfrak{p}$ is an element in $SO(3)$, which can be expressed by an orthogonal matrix, also denoted by $\mathfrak{p}$, with $\mathrm{det}\mathfrak{p}=1$. 
The components of $\mathfrak{p}$ are the coordinates of the axes $\bm{m}_i$: if we denote by $(O;\bm{e}_1,\bm{e}_2,\bm{e}_3)$ the reference frame in $\mathbb{R}^3$, then the $(i,j)$ element of $\mathfrak{p}$ is given by $\bm{e}_i\cdot\bm{m}_j$. 
We can also view $\bm{m}_j$ as functions of $\mathfrak{p}$, and use the notation $\bm{m}_j(\mathfrak{p})$ to represent the axis $\bm{m}_j$ of certain $\mathfrak{p}$.
The uniform probability measure on $SO(3)$ is denoted by $\md\mathfrak{p}$.

The operations on tensors will appear throughout the paper, so let us introduce some notations for tensors.
A $k$-th order tensor $U$ can be expressed by the basis in $\mathbb{R}^3$ as follows, 
\begin{equation}
  U=U_{j_1\ldots j_k}\bm{e}_{j_1}\otimes\ldots\otimes\bm{e}_{j_k}. 
\end{equation}
Hereafter, we adopt the Einstein convention on summation over repeated indices. 
For two tensors $U_1$ and $U_2$, we use $U_1\otimes U_2$ to represent their tensor product, where the $U_1$ components come first. 
If necessary, we write a tensor $\tsn{k}{U}$ with superscript to indicate its order. 
The dot product of two tensors with the same order is defined by
\begin{equation}
  U\cdot V= U_{j_1\ldots j_k}V_{j_1\ldots j_k}. 
\end{equation}
The Frobenius norm is then given by $\|U\|_F^2=U\cdot U$. 

Next, we define the rotation $\mathfrak{p}$ acting on a tensor. By expanding the tensor about the basis $\bm{e}_{j_1}\otimes\ldots\otimes\bm{e}_{j_k}$, the rotation is done by transforming $\bm{e}_i$ into $\bm{m}_i$, giving 
\begin{align}
  \mathfrak{p}\circ U=&U_{j_1\ldots j_k}\bm{m}_{j_1}\otimes\ldots\otimes\bm{m}_{j_k}
  .\label{rot0}
\end{align}
Since $\bm{m}_i$ can be viewed as functions of $\mathfrak{p}$, we regard $\mathfrak{p}\circ U$ as a function of $\mathfrak{p}$ and denote it as $U(\mathfrak{p})$.
We have $U(\mathfrak{p}_1\mathfrak{p}_2)=\mathfrak{p}_1\circ U(\mathfrak{p}_2)$, and
\begin{equation}
  U_1(\mathfrak{sp}_1)\cdot U_2(\mathfrak{sp}_2)=U_1(\mathfrak{p}_1)\cdot U_2(\mathfrak{p}_2), \quad\forall \mathfrak{s}\in SO(3). \label{innrot}
\end{equation}

We then introduce the notations for symmetric tensors.
For a $k$-th order tensor $U$, define its symmetrization as 
\begin{equation}
  U_{\mathrm{sym}}=\frac{1}{k!}\sum_{1\le j_1,\ldots,j_k\le 3}\sum_{\sigma}U_{j_{\sigma(1)}\ldots j_{\sigma(k)}}\bm{e}_{j_1}\otimes\ldots\otimes\bm{e}_{j_k}, \label{Usym}
\end{equation}
where the summation inside is taken over all permutations $\sigma$ of $(1,\ldots,k)$.
For any symmetric tensor $U$, its trace is defined by contracting two of the components, resulting in a $(k-2)$-th order tensor, 
\begin{align*}
  (\mathrm{tr}U)_{j_1\ldots j_{k-2}}=U_{j_1\ldots j_{k-2}ii}. 
\end{align*}
If a symmetric tensor $U$ satisfies $\mathrm{tr}U=0$, it is called a symmetric traceless tensor. 
The symmetric and traceless properties are kept under rotations. 
To express symmetric tensors, we introduce the monomial notation below, 
\begin{equation}
\bm{m}_1^{k_1}\bm{m}_2^{k_2}\bm{m}_3^{k_3}
=(\underbrace{\bm{m}_1\otimes\ldots}_{k_1}\otimes
\underbrace{\bm{m}_2\otimes\ldots}_{k_2}\otimes
\underbrace{\bm{m}_3\otimes\ldots}_{k_3}
)_{\mathrm{sym}}. \label{tensor_monomial}
\end{equation}
It is easy to see that for $k_1+k_2+k_3=k$, the tensors $\bm{m}_1^{k_1}\bm{m}_2^{k_2}\bm{m}_3^{k_3}$ give an orthogonal basis of $k$-th order symmetric tensors.
In this way, a polynomial about $\bm{m}_i$ can be regarded as a symmetric tensor, if every term in the polynomial has the same order.

We also use Kronecker delta and Levi-Civita symbol, which are given by 
$$
\delta_{ij}=\left\{
\begin{array}{ll}
  1, &i=j, \\
  0, &i\ne j.
\end{array}
\right.\quad
\epsilon_{ijk}=\left\{
\begin{array}{ll}
  1, &(ijk)=(123),(231),(312), \\
  -1, &(ijk)=(132),(213),(321),\\
  0, &\text{otherwise}.
\end{array}
\right.
$$
Two closely related tensors are the second order identity tensor, 
\begin{equation}
  \mathfrak{i}=\bm{m}_1^2+\bm{m}_2^2+\bm{m}_3^2, \nonumber
\end{equation}
and the third order determinant tensor, 
\begin{align}
  \epsilon=&\epsilon_{ijk}\bm{m}_i\otimes\bm{m}_j\otimes\bm{m}_k \nonumber\\
  =&\bm{m}_1\otimes\bm{m}_2\otimes\bm{m}_3+\bm{m}_2\otimes\bm{m}_3\otimes\bm{m}_1+\bm{m}_3\otimes\bm{m}_1\otimes\bm{m}_2\nonumber\\
  &-\bm{m}_1\otimes\bm{m}_3\otimes\bm{m}_2-\bm{m}_2\otimes\bm{m}_1\otimes\bm{m}_3-\bm{m}_3\otimes\bm{m}_2\otimes\bm{m}_1. \nonumber
\end{align}
One can verify that the above two equalities hold for any right-handed orthonormal frame $(\bm{m}_j)$. 
If $U$ is a symmetric tensor, we will use the notation
$$
\mathfrak{i}^qU=(\mathfrak{i}^q\otimes U)_{\mathrm{sym}}.
$$

To construct symmetric traceless tensors, we have the following proposition \cite{xu_tensors}. 
\begin{proposition}\label{U0}
  For each $k$-th order symmetric tensor $U$, there exists a unique $(k-2)$-th symmetric tensor $V$ such that $U-\mathfrak{i}V$ is a symmetric traceless tensor, which is denoted by $(U)_0$. The space of $k$-th order symmetric traceless tensors has the dimension $2k+1$. 
\end{proposition}
We denote an orthogonal basis of $k$-th order symmetric traceless tensors by 
\begin{equation}
  \mathbb{W}^k=\{W^k_1,\ldots,W^k_{2k+1}\}. \label{orth_basis}
\end{equation}
The following proposition is the result of group representation theory (see, for example, \cite{SpecFun}), which we will give a brief description in Appendix. 
\begin{proposition}\label{orth_basis_mat}
  The functions $W^k_i(\mathfrak{i})\cdot W^k_j(\mathfrak{p})$ for $k=0,1,\ldots$ give a complete orthogonal basis of $L^2(SO(3))$. 
\end{proposition}

Next, we state the approximation result by the functions $w^k_{ij}(\mathfrak{p})=W^k_i(\mathfrak{i})\cdot W^k_j(\mathfrak{p})$ \cite{wigner_SXZ}. 
To this end, we need to introduce the $H^{s}$ norms on $SO(3)$. 
The definition is similar to the case in $\mathbb{R}^3$, by substituting the three derivatives $\partial_i$ with the three differential operators $\mathcal{L}_i\,(i=1,2,3)$ on $SO(3)$. 
They represent the infinitesimal rotations round $\bm{m}_i$, whose explicit formulae can be found in Appendix. 
The $H^s$ norm is denoted by $\|\cdot\|_{H^s}$, and the $H^s$ semi-norm is denoted by $|\cdot|_{H^s}$.
If there is no subscript, $\|\cdot\|$ denotes the $L^2$ norm. 

Define the projection operator $\pi_N$ as:
\begin{equation}
  \pi_Ng=\sum_{k\le N}\lambda^k_{ij}w^k_{ij}(\mathfrak{p}), \text{ such that }
\int \big(g(\mathfrak{p})-\pi_Ng(\mathfrak{p})\big)w^k_{ij}(\mathfrak{p})\md\mathfrak{p}=0. \label{projscal}
\end{equation}
\begin{proposition}\label{approxbas}
  For the function $g(\mathfrak{p})\in H^s(SO(3))$, we have 
  \begin{equation}
    \|g-\pi_Ng\|\le CN^{-s}|g|_{H^s}, 
  \end{equation}
  where $C$ is a constant. 
\end{proposition}

The notation of norms can be extended to tensor-valued functions about multiple $\mathfrak{p}$-variables: 
for two tensor-valued functions $A(\mathfrak{p}_1,\ldots,\mathfrak{p}_l)$ and $B(\mathfrak{p}_1,\ldots,\mathfrak{p}_l)$, if $A$ and $B$ have the same order, we define the inner product as 
\begin{align}
  (A,B)=\int A(\mathfrak{p}_1,\ldots,\mathfrak{p}_l)\cdot B(\mathfrak{p}_1,\ldots,\mathfrak{p}_l) \,\md\mathfrak{p}_1\ldots\md \mathfrak{p}_l. \label{innp_tensor}
\end{align}
The $L^2$ norm is then defined as 
\begin{align}
  \|A\|^2=(A,A)=\int \|A(\mathfrak{p}_1,\ldots,\mathfrak{p}_l)\|_F^2 \,\md\mathfrak{p}_1\ldots\md \mathfrak{p}_l. \label{norm_tensor}
\end{align}
The $H^s$ norms are given similarly. 

\section{General form of expansion\label{expansion_general}}
In this section, we discuss the expansion of interaction kernels in the general case, i.e. without considering molecular symmetry. 
First, we introduce the gradient expansion to decouple the spatial variables, which has been used previously to deal with various systems without orientational variables.
Then, starting from the formula after the gradient expansion is done, we discuss how to deal with the orientational variables. 

\subsection{Molecular model and gradient expansion}
Based on microscopic potential and statistical mechanics, one could write down a molecular model, which may involve various approaches such as mean-field theory or cluster expansion \cite{Mayer_book,Landau_Lifshitz_Statistical}. 
No matter what approaches are used, the free energy typically takes a form including the contribution of local entropy term and nonlocal interactions of molecule clusters of two, three, four, and so on. 
It is written as 
\begin{align}
{\beta_0}{\mathscr{F}[f]}=&
\int \md\bm{x} \md \mathfrak{p}f(\bm{x},\mathfrak{p})\ln f(\bm{x},\mathfrak{p})+\mathscr{F}_2+\mathscr{F}_3+\mathscr{F}_4+\ldots ,\nonumber
\end{align}
where the nonlocal interactions terms are given by 
\begin{align}
\mathscr{F}_2=&\frac{1}{2!}\int\md\bm{x}_1\md \mathfrak{p}_1\md\bm{x}_2\md \mathfrak{p}_2\mathscr{G}_2(\bm{r}_2,\mathfrak{p}_1,\mathfrak{p}_2)f(\bm{x}_1,\mathfrak{p}_1)f(\bm{x}_2,\mathfrak{p}_2),\nonumber\\
\mathscr{F}_3=&\frac{1}{3!}\int\md\bm{x}_1\md \mathfrak{p}_1\md\bm{x}_2\md \mathfrak{p}_2\md\bm{x}_3\md\mathfrak{p}_3\mathscr{G}_3(\bm{r}_2,\bm{r}_3,\mathfrak{p}_1,\mathfrak{p}_2,\mathfrak{p}_3)f(\bm{x}_1,\mathfrak{p}_1)f(\bm{x}_2,\mathfrak{p}_2)f(\bm{x}_3,\mathfrak{p}_3),\nonumber\\
\mathscr{F}_4=&\frac{1}{4!}\int\md\bm{x}_1\md \mathfrak{p}_1\md\bm{x}_2\md \mathfrak{p}_2\md\bm{x}_3\md\mathfrak{p}_3\md\bm{x}_4\md\mathfrak{p}_4\mathscr{G}_4(\bm{r}_2,\bm{r}_3,\bm{r}_4,\mathfrak{p}_1,\mathfrak{p}_2,\mathfrak{p}_3,\mathfrak{p}_4)\nonumber\\
&\qquad\qquad f(\bm{x}_1,\mathfrak{p}_1)f(\bm{x}_2,\mathfrak{p}_2)f(\bm{x}_3,\mathfrak{p}_3)f(\bm{x}_4,\mathfrak{p}_4). \label{virial}
\end{align}
Here, $(\bm{x}_j,\mathfrak{p}_j)$ represents the position and orientation of the molecule $j$, $\bm{r}_j=\bm{x}_j-\bm{x}_1$ is the relative position to the molecule $1$, 
and $\beta_0$ is the inverse of the product of the Boltzmann constant and the absolute temperature. 



The entropy is a local term that can be handled in different ways.
One possible approach is to use the so-called Bingham closure for rod-like molecules \cite{bingham1974antipodally,ball2010nematic,RodModel}, which is also adopted for bent-core molecules \cite{BentModel}. 
This apporach can always be carried out if we are able to deal with the nonlocal interaction terms. 
Therefore, we do not discuss this term in this paper. 

Our focus is the nonlocal interaction terms $\mathscr{F}_l$. 
The interaction kernels $\mathscr{G}_l$ are functions of the molecular potential that might involve numerous types of forces, which we will not try to specify. 
The nonlocal interaction terms $\mathscr{F}_l$ are typically truncated somewhere. 
As an example, if the concentration is low, it would suffice to keep the $\mathscr{F}_2$ term only.
In this case, one could use $\mathscr{G}_2=1-\exp\big(-\beta_0\mathscr{U}(\bm{r}_2,\mathfrak{p}_1,\mathfrak{p}_2)\big)$ where $\mathscr{U}$ is the potential for a pair of molecules. 

We start from doing Taylor expansions on $f(\bm{x}_j,\mathfrak{p}_j)=f(\bm{x}_1+\bm{r}_j,\mathfrak{p}_j)$ about $\bm{r}_j$, leading to 
\begin{align}
{\mathscr{F}_l[f]}=
  \sum_{k_2,\ldots,k_l}\frac{1}{l!k_2!\ldots k_l!}\int&\md\bm{x}_1\md\mathfrak{p}_1\ldots\md \mathfrak{p}_l
  f(\bm{x}_1,\mathfrak{p}_1)\nonumber\\
  &\mathscr{M}_l^{k_2,\ldots,k_l}(\mathfrak{p}_1\ldots,\mathfrak{p}_l)\cdot\nabla^{k_2} f(\bm{x}_1,\mathfrak{p}_2)\otimes\ldots\otimes\nabla^{k_l} f(\bm{x}_1,\mathfrak{p}_l),\label{TlExp}
\end{align}
where we define the tensors $\mathscr{M}_l^{k_2,\ldots,k_l}$ as follows, 
\begin{equation}
  \mathscr{M}_l^{k_2,\ldots,k_l}(\mathfrak{p}_1,\ldots,\mathfrak{p}_l)=\int \mathscr{G}_l(\bm{r}_2,\ldots,\bm{r}_l,\mathfrak{p}_1,\ldots,\mathfrak{p}_l)
  \bm{r}_2^{k_2}\otimes\ldots\otimes\bm{r}_l^{k_l}\,\md\bm{r}_2\ldots \md\bm{r}_l. \label{spcmoment}
\end{equation}
%
%
%
In the above, we have effectively done the gradient expansion. It has been adopted in many systems where no orientational variables are involved (such as \cite{Cahn_Hilliard_1}), where $\mathscr{M}_l^{k_2,\ldots,k_l}$ are constant tensors, so that $\mathscr{F}$ becomes a functional about $f$ and their derivatives. 
In some systems, such a manipulation is done in the Fourier space, where the expansion is done about Fourier modes.
It leads to polynomials of Fourier modes \cite{LB,fredrickson1987fluctuation}, which is formally equivalent to the gradient expansion by Fourier transformations.
It certainly requires some conditions for the gradient expansion to be appropriate. 
In this work, however, we assume its appropriateness and start our discussion from \eqref{TlExp} and \eqref{spcmoment}. 

The focus of this paper is the expansion of $\mathscr{M}_l^{k_2,\ldots,k_l}$ about the orientational variables $\mathfrak{p}_i$. 
After the expansion, the variables $\mathfrak{p}_i$ are separated, so that the integrals $\int\md\mathfrak{p}_i$ can be decoupled. 
The interaction terms $\mathscr{F}_l$ then become functionals about several quantities averaged by $f(\mathfrak{p})$, denoted by $\langle h\rangle$ that we define as 
\begin{equation}
  \langle h\rangle=\int h(\mathfrak{p})f(\bm{x},\mathfrak{p})\,\md\mathfrak{p}. 
\end{equation}
The average is taken over $SO(3)$, so that $\langle h\rangle$ is a function of $\bm{x}$. 
The expansion shall satisfy several symmetry arguments, which we will discuss throughout the rest of paper. 
We shall discuss $\mathscr{M}_2^k$ in a detailed manner to clearly illustrate the principles. 
Explicit expressions will be given for the interaction terms up to $\mathscr{F}_4$. 
As indicated by the Landau-de Gennes theory, this is expected to cover most applications.


\subsection{Expansion of $\mathscr{M}_2^k$}

We begin with writing down the symmetry arguments that the kernel function shall satisfy. 
Regardless of the molecular potential, the interaction between a pair of molecules shall only depend on their relative position and orientation, and be invariant when two molecules interchange. 
This leads to natural symmetries in the kernel function $\mathscr{G}_2$, 
given by 
\begin{align}
  \mathscr{G}_2(\mathfrak{t}\bm{r}_2,\mathfrak{tp}_1,\mathfrak{tp}_2)&=\mathscr{G}_2(\bm{r}_2,\mathfrak{p}_1,\mathfrak{p}_2), \quad \forall \mathfrak{t}\in SO(3), \label{rotate0}\\
  \mathscr{G}_2(-\bm{r}_2,\mathfrak{p}_2,\mathfrak{p}_1)&=\mathscr{G}_2(\bm{r}_2,\mathfrak{p}_1,\mathfrak{p}_2). \label{switch00}
\end{align}

We first seek the expansion consistent with \eqref{rotate0}.
It yields
\begin{align*}
  \mathscr{M}_2^k(\mathfrak{tp}_1,\mathfrak{tp}_2)=&\int \bm{r}_2^k\mathscr{G}_2(\bm{r}_2,\mathfrak{tp}_1,\mathfrak{tp}_2) \md\bm{r}_2
  =\int (\mathfrak{t}\bm{r}_2)^k\mathscr{G}_2(\mathfrak{t}\bm{r}_2,\mathfrak{tp}_1,\mathfrak{tp}_2) \md(\mathfrak{t}\bm{r}_2)\\
  =&\int (\mathfrak{t}\bm{r}_2)^k\mathscr{G}_2(\bm{r}_2,\mathfrak{p}_1,\mathfrak{p}_2) \md\bm{r}_2
  =\mathfrak{t}\circ \int \bm{r}_2^k\mathscr{G}_2(\bm{r}_2,\mathfrak{p}_1,\mathfrak{p}_2) \md\bm{r}_2\\
  =&\mathfrak{t}\circ \mathscr{M}_2^k(\mathfrak{p}_1,\mathfrak{p}_2). 
\end{align*}
Since $\mathscr{M}_2^k$ is a $k$-th order tensor, let us express it in the basis $\bm{m}_i(\mathfrak{p}_1)$, 
$$
\mathscr{M}_2^k=\sum_{k_1+k_2+k_3=k}\big(\mathscr{M}_2^k\cdot\bm{m}_1^{k_1}(\mathfrak{p}_1)\bm{m}_2^{k_2}(\mathfrak{p}_1)\bm{m}_3^{k_3}(\mathfrak{p}_1)\big)
\frac{k!}{k_1!k_2!k_3!}\bm{m}_1^{k_1}(\mathfrak{p}_1)\bm{m}_2^{k_2}(\mathfrak{p}_1)\bm{m}_3^{k_3}(\mathfrak{p}_1). 
$$
Notice that the coefficients $\mathscr{M}_2^k\cdot\bm{m}_1^{k_1}(\mathfrak{p}_1)\bm{m}_2^{k_2}(\mathfrak{p}_1)\bm{m}_3^{k_3}(\mathfrak{p}_1)$ are scalar functions of $\mathfrak{p}_1^{-1}\mathfrak{p}_2$, because we have the following, 
\begin{align*}
  \mathscr{M}_2^k(\mathfrak{tp}_1,&\mathfrak{tp}_2)\cdot\bm{m}_1^{k_1}(\mathfrak{tp}_1)\bm{m}_2^{k_2}(\mathfrak{tp}_1)\bm{m}_3^{k_3}(\mathfrak{tp}_1)\\
  =&\Big(\mathfrak{t}\circ \mathscr{M}_2^k(\mathfrak{p}_1,\mathfrak{p}_2)\Big)\cdot \mathfrak{t}\circ\Big(\bm{m}_1^{k_1}(\mathfrak{p}_1)\bm{m}_2^{k_2}(\mathfrak{p}_1)\bm{m}_3^{k_3}(\mathfrak{p}_1)\Big)\\
  =&\mathscr{M}_2^k(\mathfrak{p}_1,\mathfrak{p}_2)\cdot\bm{m}_1^{k_1}(\mathfrak{p}_1)\bm{m}_2^{k_2}(\mathfrak{p}_1)\bm{m}_3^{k_3}(\mathfrak{p}_1)\\
 (\text{let }\mathfrak{t}=\mathfrak{p}_1^{-1})\quad =&\mathscr{M}_2^k(\mathfrak{i},\mathfrak{p}_1^{-1}\mathfrak{p}_2)\cdot\bm{e}_1^{k_1}\bm{e}_2^{k_2}\bm{e}_3^{k_3}. 
\end{align*}
So, we expand this scalar by the orthogonal basis given in Proposition \ref{orth_basis_mat}. 
It can be written as the sum of some terms given by $V_1^m(\mathfrak{i})\cdot V^m(\mathfrak{p}_1^{-1}\mathfrak{p}_2)$,
where $V_1^m$ and $V^m$ are two symmetric traceless tensors of $m$-th order. 
Plugging it into $\mathscr{M}_2^k$, we know that $\mathscr{M}_2^k$ can be expanded into the sum of terms of the following form, 
\begin{align}
  &Y_1^k(\mathfrak{p}_1)_{i_1\ldots i_k}\big(V_1^m(\mathfrak{i})\cdot V^m(\mathfrak{p}_1^{-1}\mathfrak{p}_2)\big)=Y_1^k(\mathfrak{p}_1)_{i_1\ldots i_k}\big(V_1^m(\mathfrak{p}_1)\cdot V^m(\mathfrak{p}_2)\big)\nonumber\\
  =&Y(\mathfrak{p}_1)_{i_1\ldots i_kj_1\ldots j_m}V^m(\mathfrak{p}_2)_{j_1\ldots j_m}, \label{term0}
\end{align}
where $Y_1^k$ is a $k$-th order symmetric tensor, and we denote $Y=Y_1^k\otimes V_1^m$. 
The above expansion is already variable separated: 
we could take it back into the Taylor's expansion \eqref{TlExp} and obtain 
\begin{align}
  &\int f(\bm{x},\mathfrak{p}_1)Y(\mathfrak{p}_1)_{i_1\ldots i_kj_1\ldots j_m}V^m(\mathfrak{p}_2)_{j_1\ldots j_m}\cdot \nabla^k f(\bm{x},\mathfrak{p}_2)\, \md\bm{x}\md\mathfrak{p}_1\md \mathfrak{p}_2\nonumber\\
  =&\int \left(\int Y(\mathfrak{p}_1)_{i_1\ldots i_kj_1\ldots j_m}f(\bm{x},\mathfrak{p}_1) \md\mathfrak{p}_1\right)
  \partial_{i_1\ldots i_k}\left( \int V^m(\mathfrak{p}_2)_{j_1\ldots j_m}f(\bm{x},\mathfrak{p}_2)\,\md \mathfrak{p}_2\right)\,\md\bm{x}\nonumber\\
  =&\int \langle Y\rangle_{i_1\ldots i_kj_1\ldots j_m}\partial_{i_1\ldots i_k}\langle V^m\rangle_{j_1\ldots j_m}\,\md\bm{x}. \label{sepvar0}
\end{align}
We can see that the last expression is already a term about two averaged tensors $\langle Y\rangle$ and $\langle V^m\rangle$, which is our motivation of doing the expansion. 

When $Y_1^k$ takes the basis tensors $\bm{m}_1^{k_1}\bm{m}_2^{k_2}\bm{m}_3^{k_3}$ and $V_1^m$, $V^m$ take the basis tensors in $\mathbb{W}^m$, the terms in \eqref{term0} are linearly independent. Note that $Y_1^k$ has ${k+2\choose 2}$ choices, $V_1^m$ and $V^m$ both have $2m+1$ choices. Thus, the total number of these terms is ${k+2\choose 2}(2m+1)^2$. 
Furthermore, from Proposition \ref{orth_basis_mat} we know that when $m$ takes all nonnegative integers, these terms form a complete orthogonal basis. 

However, the above form is inconvenient when discussing \eqref{switch00} and molecular symmetries afterwards. 
In what follows, we decompose $Y$ into symmetric traceless tensors and try to identify the linearly independent terms after the decomposition. 

The decomposition of a tensor into symmetric traceless tensors is described briefly in Appendix.
Here, we only present the result: for an $r$-th order tensor $X$, we have 
\begin{align}
  X_{j_1\ldots j_r}=\sum_{\substack{0\le s\le r, s\text{ even}\\ \{\tau_1,\ldots,\tau_s\}\cup\{\sigma_1,\ldots\sigma_{r-s}\}\\=\{1,\ldots,r\}}}
  &\delta_{j_{\tau_1}j_{\tau_2}}\ldots\delta_{j_{\tau_{s-1}}j_{\tau_s}}U_{j_{\sigma_1}\ldots j_{\sigma_{r-s}}}\nonumber\\
  &+\epsilon_{j_{\tau_1}j_{\tau_2}\nu}\delta_{j_{\tau_3}j_{\tau_4}}\ldots\delta_{j_{\tau_{s-1}}j_{\tau_s}}U_{\nu j_{\sigma_1}\ldots j_{\sigma_{r-s}}}\nonumber\\
  +\sum_{\substack{0\le s\le r, s\text{ odd}\\ \{\tau_1,\ldots,\tau_s\}\cup\{\sigma_1,\ldots\sigma_{r-s}\}\\=\{1,\ldots,r\}}}&\epsilon_{j_{\tau_1}j_{\tau_2}j_{\tau_3}}\delta_{j_{\tau_4}j_{\tau_5}}\ldots\delta_{j_{\tau_{s-1}}j_{\tau_s}}U_{j_{\sigma_1}\ldots j_{\sigma_{r-s}}}. \label{symtrlsdecomp}
\end{align}
Here, we use $U$ to denote any symmetric traceless tensor (in different terms $U$ can be different). 

Now we apply the above decomposition to the tensor $Y=Y_1^k\otimes V_1^m$ in \eqref{term0}. 
We shall keep in mind that $Y_1^k$ is a $k$-th order symmetric tensor, and $V_1^m$, $V^m$ are $m$-th order symmetric traceless tensors. 
Let us discuss indices of $\delta$ and $\epsilon$ in the decomposition. 
Note that in a symmetric tensor, every index is equivalent. 
Thus, we only need to examine how many of the indices are located in $Y_1^k$ or $V_1^m$. 
\begin{enumerate}
\item If both indices in a $\delta$ are located in $V_1^m$, the resulting term is zero when taking into \eqref{term0}, because it leads to the contraction of two indices in $V^m$. 
\item Both indices in a $\delta$ are located $Y_1^k$. Suppose the number of such $\delta$ is $q$. 
\item One index in a $\delta$ is located in $Y_1^k$, while the other is located in $V_1^m$. Suppose the number of such $\delta$ is $p$. 
\item For the three indices in $\epsilon$, if any two of them are located in $Y_1^k$ (or $V_1^m$), then the term will vanish because $Y_1^k$ and $V_1^m$ are symmetric. The non-vanishing term must have one index in $Y_1^k$, one in $V_1^m$, and the third can only be the $\nu$ of $U_{\nu\ldots}$ in \eqref{symtrlsdecomp}. In other words, the second summation in \eqref{symtrlsdecomp} contributes nothing in \eqref{term0}. 
\end{enumerate}
Summarizing these cases, we obtain some terms expressed by a pair of symmetric traceless tensors.
If we label the tensor order of $U$ in \eqref{symtrlsdecomp} by $U^r$, these terms can be written as 
\begin{subequations}\label{symtrlsinner}
\begin{align}
  &\left(\delta_{j_1j_2}\ldots \delta_{j_{2q-1}j_{2q}}U^r(\mathfrak{p}_1)_{j'_1\ldots j'_{r-p} i_1\ldots i_p}V^m(\mathfrak{p}_2)_{j''_{1}\ldots j''_{m-p}i_1\ldots i_p}\right)_{\mathrm{sym}}, \\
  &\left(\epsilon_{\zeta_1\zeta_2\nu}\delta_{j_1j_2}\ldots \delta_{j_{2q-1}j_{2q}}U^r(\mathfrak{p}_1)_{\zeta_1j'_{1}\ldots j'_{r-p-1}i_1\ldots i_p}V^m(\mathfrak{p}_2)_{\zeta_2j''_{1}\ldots j''_{m-p-1}i_1\ldots i_p}\right)_{\mathrm{sym}}. 
\end{align}
\end{subequations}
They are all symmetrized since $Y_1^k$ is symmetric. 
Let us define two short notations, 
\begin{subequations}\label{oversetdef}
\begin{align}
  &U^r\overset{p}{\cdot} V^m=(U^r_{i_1\ldots i_pj_1\ldots j_{r-p}}V^m_{i_1\ldots i_pj'_1\ldots j'_{m-p}})_{\mathrm{sym}}, \\
  &U^r\overset{p}\times V^m=(\epsilon_{\zeta_1\zeta_2 \nu} U^r_{\zeta_1i_1\ldots i_pj_1\ldots j_{r-p-1}}V^m_{\zeta_2i_1\ldots i_pj'_1\ldots j'_{m-p-1}})_{\mathrm{sym}}. 
\end{align}
\end{subequations}
Rewriting the terms in \eqref{symtrlsinner} and noticing the relation of tensor order, any term given by \eqref{term0} can be expressed linearly by the following terms, 
\begin{subequations}\label{terms}
\begin{align}
  &\idn^qU^r(\mathfrak{p}_1)\overset{p}{\cdot}V^m(\mathfrak{p}_2),\quad k=2q+r+m-2p,\label{terms:a}\\
  &\idn^qU^r(\mathfrak{p}_1)\overset{p}{\times} V^m(\mathfrak{p}_2),\quad k=2q+r+m-2p-1.\label{terms:b}
\end{align}
\end{subequations}
%
Note that $V^m$ is the same tensor in \eqref{term0} and \eqref{terms}. 
We shall prove the following. 

\begin{theorem}\label{twoforms}
  Let $k$ and $V^m$ be fixed. Let $r,p,q$ vary and $U^r\in \mathbb{W}^r$. The terms given in \eqref{terms} are linearly independent, and are linearly equivalent to the terms given in \eqref{term0}. 
\end{theorem}
\begin{proof}
In the above derivation, we actually show that \eqref{term0} can be linearly expressed by the terms given in \eqref{terms}. 
Recall that when $V^m$ is fixed, the total number of linearly independent terms given by \eqref{term0} is $\frac{1}{2}(k+2)(k+1)\cdot (2m+1)$. 
Thus, we only need to prove that the number of terms given by \eqref{terms} (when $U^r\in\mathbb{W}^r$) equals to this value. 


The $r$-th order symmetric traceless tensor $U^r$ has $2r+1$ choices.
We shall count the choices of $(p,q)$ with a fixed $r$. 
In \eqref{terms}, the indices shall be nonnegative integers.
Moreover, we require $p\le r,m$ in \eqref{terms:a} and $p+1\le r,m$ in \eqref{terms:b}. Thus, we deduce the range for the indices, 
\begin{align*}
  \eqref{terms:a}:\quad & r+m-k\text{ even, }\max\{0,\frac{r+m-k}{2}\}\le p\le \min\{m,r\};\\
  \eqref{terms:b}:\quad & r+m-k\text{ odd, }\max\{0,\frac{r+m-k-1}{2}\}\le p\le \min\{m,r\}-1. 
\end{align*}
If $k\le m$, by requiring the upper bounds no less than the lower bounds in \eqref{terms}, we deduce that $m-k\le r\le m+k$. The number of $p$ available is given by 
\begin{align*}
  \eqref{terms:a}:\quad & \frac{k-|r-m|}{2}+1,\\
  \eqref{terms:b}:\quad & \frac{k-|r-m|-1}{2}+1.
\end{align*}
Hence, the total number of terms is 
\begin{align*}
  &\sum_{k-|r-m|\ge 0\ \mathrm{even}} \frac{k-|r-m|+2}{2}(2r+1)+\sum_{k-|r-m|\ge 0\ \mathrm{odd}} \frac{k-|r-m|+1}{2}(2r+1). 
\end{align*}
Let $u=r-m$ so that $r=m+u$. The above number is calculated as 
\begin{align*}
  &\sum_{k-|u|\ge 0\ \mathrm{even}} \frac{k-|u|+2}{2}(2m+1+2u)+\sum_{k-|u|\ge 0\ \mathrm{odd}} \frac{k-|u|+1}{2}(2m+1+2u)\\
  =&\sum_{k-|u|\ge 0\ \mathrm{even}} \frac{k-|u|+2}{2}(2m+1)+\sum_{k-|u|\ge 0\ \mathrm{odd}} \frac{k-|u|+1}{2}(2m+1)\\
  =&(2m+1)\Big(\sum_{k-|u|\ge 0\ \mathrm{even}} \frac{k-|u|+2}{2}+\sum_{k-|u|\ge 0\ \mathrm{odd}} \frac{k-|u|+1}{2}\Big)\\
  =&(2m+1)\cdot \frac{1}{2}(k+1)(k+2). 
\end{align*}

If $k>m$, let us do induction about $k$, based on $k=m,m-1$ that have been shown above. 
Suppose that for $k-2$, the total number is $\frac{1}{2}k(k-1)(2m+1)$. 
If $q>0$, a term for $k$ corresponds to a term for $k-2$ by substituting $q$ with $q-1$. 
Now we count the number of terms where $q=0$. 
There are two cases:
\begin{itemize}
\item $r=k-m+2p$ for $0\le p\le m$;
\item $r=k-m+2p+1$ for $0\le p\le m-1$.
\end{itemize}
Summarizing the two cases, we have $k-m\le r\le k+m$ and $p$ is determined correspondingly. The total number of terms when $q=0$ is thus 
$$
\sum_{r=k-m}^{k+m}(2r+1)=(2k+1)(2m+1)=\frac{1}{2}(k+2)(k+1)(2m+1)-\frac{1}{2}k(k-1)(2m+1). 
$$

The only case remaining is $m=0$, $k=1$, for which we can count directly. 
\end{proof}
\textit{Remark.} It follows from the paragraph below \eqref{sepvar0} that when $(U^r,V^m)\in\mathbb{W}^r\times \mathbb{W}^m$ with both $r$ and $m$ varying, the terms in \eqref{terms} are linearly independent. 
We shall point out that it is not always the case that the terms expressed by symmetric traceless tensors are linearly independent, as we will see in the expansion of $\mathscr{M}_4^{0,0,0}$ afterwards. In that case, the approach in the proof above is also useful. 


We next deal with the property \eqref{switch00}. From the definiton of $\mathscr{M}_2^k$, it leads to
\begin{align}
  \mathscr{M}_2^k(\mathfrak{p}_2,\mathfrak{p}_1)=(-1)^k\mathscr{M}_2^k(\mathfrak{p}_1,\mathfrak{p}_2). \label{switch0}
\end{align}
For the terms in \eqref{terms} where $(U^r,V^m)\in\mathbb{W}^r\times \mathbb{W}^m$, we consider the following two sets, 
\begin{subequations}\label{subspace_switch}
\begin{align}
  \mathbb{X}^k_{n,+1}=&\big\{\mathfrak{i}^qU^r(\mathfrak{p}_1)\overset{p}{\cdot}V^m(\mathfrak{p}_2)+\mathfrak{i}^qV^m(\mathfrak{p}_1)\overset{p}{\cdot}U^r(\mathfrak{p}_2):\ k=2q+r+m-2p,\ r,m\le n\big\}\nonumber\\
  &\cup
  \big\{\mathfrak{i}^qU^r(\mathfrak{p}_1)\overset{p}{\times} V^m(\mathfrak{p}_2)-\mathfrak{i}^qV^m(\mathfrak{p}_1)\overset{p}{\times} U^r(\mathfrak{p}_2):\ k=2q+r+m-2p-1,\ r,m\le n,\ U^r\ne V^m\big\},\label{subspace_switch:a}\\
  \mathbb{X}^{k}_{n,-1}=&\big\{\mathfrak{i}^qU^r(\mathfrak{p}_1)\overset{p}{\cdot}V^m(\mathfrak{p}_2)-\mathfrak{i}^qV^m(\mathfrak{p}_1)\overset{p}{\cdot}U^r(\mathfrak{p}_2):\ k=2q+r+m-2p,\ r,m\le n,\ U^r\ne V^m\big\}\nonumber\\
  &\cup
  \big\{\mathfrak{i}^qU^r(\mathfrak{p}_1)\overset{p}{\times} V^m(\mathfrak{p}_2)+\mathfrak{i}^qV^m(\mathfrak{p}_1)\overset{p}{\times} U^r(\mathfrak{p}_2):\ k=2q+r+m-2p-1,\ r,m\le n\big\}.\label{subspace_switch:b} 
\end{align}
\end{subequations}
%
Here, we require $U^r\ne V^m$ in some sets to avoid zero. 
The terms in $\mathbb{X}^k_{n,1}\cup\mathbb{X}^k_{n,-1}$ are also linearly independent and linearly equivalent to \eqref{terms}. 

Any term $A(\mathfrak{p}_1,\mathfrak{p}_2)\in \mathbb{X}^{k}_{n,\pm 1}$ satisfies $A(\mathfrak{p}_2,\mathfrak{p}_1)=\pm A(\mathfrak{p}_1,\mathfrak{p}_2)$. 
Thus, it is easy to verify that the two spaces are orthogonal using the definition \eqref{innp_tensor}. 
It follows from \eqref{switch0} that $(\mathscr{M}_2^k, A)=0$ for any $A(\mathfrak{p}_1,\mathfrak{p}_2)\in \mathbb{X}^{k}_{n,(-1)^{k+1}}$. 
Therefore, for odd $k$, the expansion of $\mathscr{M}_2^k$ can only have terms in $\mathbb{X}^{k}_{n,-1}$, while for even $k$ it can only have terms in $\mathbb{X}^{k}_{n,+1}$. 


\begin{corollary}
  When $(U^r,V^m)\in\mathbb{W}^r\times\mathbb{W}^m$ for $r,m\le n$, the terms in $\mathbb{X}^{k}_{n,(-1)^k}$ satisfy \eqref{rotate0} and \eqref{switch00}, and are linearly independent. 
\end{corollary}

In the following, we write down orthogonal basis of $\mathrm{span}\mathbb{X}^k_{n,\pm 1}$. 
Note that for a pair of tensors $(U^r,V^m)$, there can be multiple terms in \eqref{terms} involving them. 
To achieve orthogonality for these terms, we derive some symmetric traceless tensors related to these terms. 
We know from Proposition \ref{U0} that there exists a unique symmetric traceless tensor generated by $U^r\overset{p}{\cdot}V^m$. 
Below, we would like to derive the explicit formulae. 
Consider 
\begin{align}
  (U^r\overset{p}{\cdot}V^m)_0=U^r\overset{p}{\cdot}V^m+\sum_{l=1}^{\min\{r,m\}-p}a^{r,m,p}_{l}\mathfrak{i}^{l}U\overset{p+l}{\cdot}V. \label{UVtrls1}
\end{align}
Calculating the trace using \eqref{Usym}, we deduce that
\begin{align}
  \mathrm{tr}\Big(\mathfrak{i}^{l}U^r\overset{p+l}{\cdot}V^m\Big)=&2l\big(2(r+m-2p)+1-2l\big)\mathfrak{i}^{l-1}U^r\overset{p+l}{\cdot}V^m\nonumber\\
  &+2(r-p-l)(m-p-l)\mathfrak{i}^{l}U^r\overset{p+l+1}{\cdot}V^m. \nonumber
\end{align}
So we have 
\begin{align}
  2l\big(2(r+m-2p)+1-2l\big)a^{r,m,p}_{l}+2(r-p-l+1)(m-p-l+1)a^{r,m,p}_{l-1}=0. \nonumber
\end{align}
Therefore,
\begin{align}
  a^{r,m,p}_{l}=(-1)^l\frac{(r-p)!(m-p)!\big(2(r+m-2p)-1-2l\big)!!}{l!(r-p-l)!(m-p-l)!\big(2(r+m-2p)-1\big)!!}. \nonumber
\end{align}
Similarly, we deduce that 
\begin{align}
  (U^r\overset{p}{\times}V^m)_0=U^r\overset{p}{\times}V^m+\sum_{l=1}^{\min\{r,m\}-p-1}b^{r,m,p}_{l}\mathfrak{i}^{l}U^r\overset{p+l}{\times}V^m, \label{UVtrls2}
\end{align}
where the coefficients are
$$
b^{r,m,p}_{l}=(-1)^l\frac{(r-p-1)!(m-p-1)!\big(2(r+m-2p)-3-2l\big)!!}{l!(r-p-l-1)!(m-p-l-1)!\big(2(r+m-2p)-3\big)!!}. 
$$

\begin{theorem}\label{orth_M2}
  The following terms give an orthogonal basis of the space $\mathrm{span}\mathbb{X}^{k}_{n,(-1)^k}$: 
  \begin{subequations}\label{orthterms}
  \begin{align}
    &\mathfrak{i}^q\Big(U^r(\mathfrak{p}_1)\overset{p}{\cdot}V^m(\mathfrak{p}_2)+(-1)^kV^m(\mathfrak{p}_1)\overset{p}{\cdot}U^r(\mathfrak{p}_2)\Big)_0,\quad k=2q+r+m-2p,\label{orthterms:a}\\
    &\mathfrak{i}^q\Big(U^r(\mathfrak{p}_1)\overset{p}{\times} V^m(\mathfrak{p}_2)-(-1)^kV^m(\mathfrak{p}_1)\overset{p}{\times} U^r(\mathfrak{p}_2)\Big)_0,\quad k=2q+r+m-2p-1.\label{orthterms:b}
  \end{align}
  where $U^r$ and $V^m$ take symmetric traceless tensors in $\mathbb{W}^r$ and $\mathbb{W}^m$, respectively, for $r,m\le n$. 
  \end{subequations}
\end{theorem}
\begin{proof}
  The expressions \eqref{UVtrls1} and \eqref{UVtrls2} indicate that the terms in \eqref{orthterms} are linearly equivalent to $\mathbb{X}^k_{n,(-1)^k}$, and are linarly independent by since the number of terms does not change. So, we only need to show orthogonality. 
  
  First, we need to notice that if $r\ne m$ or $j\ne l$, then we have 
  \begin{align}
    \int (W^r_j(\mathfrak{p})\otimes W^m_l(\mathfrak{p}))_{i_1\ldots i_ri_1'\ldots i_m'} \,\md\mathfrak{p}=\int (W^r_j(\mathfrak{p}))_{i_1\ldots i_r}(W^m_l(\mathfrak{p}))_{i_1'\ldots i_m'} \,\md\mathfrak{p}=0. \label{orth_compon}
  \end{align}
  Here, we could write 
  \begin{align*}
    (W^r_j(\mathfrak{p}))_{i_1\ldots i_r}=&W^r_j(\mathfrak{p})\cdot \bm{e}_{i_1}\ldots\bm{e}_{i_r}
    =W^r_j(\mathfrak{p})\cdot (\bm{e}_{i_1}\ldots\bm{e}_{i_r})_0\nonumber\\
    =&\sum_{j'}W^r_j(\mathfrak{p})\cdot \lambda_{j'}W^r_{j'}(\mathfrak{i}), 
  \end{align*}
  where in the second equality we use the fact that $W^r_j$ is symmetric traceless, and in the last equality we express $(\bm{e}_{i_1}\ldots\bm{e}_{i_r})_0$ by the basis in $\mathbb{W}^r$. 
  The equation \eqref{orth_compon} then comes directly from the orthogonality in Proposition \ref{orth_basis_mat}. 
  
  To deal with the terms in \eqref{orthterms}, we go back to \eqref{terms}. 
  Let us denote in short a term in \eqref{terms} by $\Phi(U^r(\mathfrak{p}_1),V^m(\mathfrak{p}_2))$. 
  Consider two terms in \eqref{terms}, $\Phi_1(U_1^{r_1}(\mathfrak{p}_1),V_1^{m_1}(\mathfrak{p}_2))$ and $\Phi_2(U_2^{r_2}(\mathfrak{p}_1),V_2^{m_2}(\mathfrak{p}_2))$, where $(U_i^{r_i},V_i^{m_i})\in\mathbb{W}^{r_i}\times\mathbb{W}^{m_i}$. 
  If $U_1^{r_1}\ne U_2^{r_2}$ or $V_1^{m_1}\ne V_2^{r_2}$, we show the orthogonality using \eqref{orth_compon}. 
  To recognize this, we notice that a dot product of two tensors can be rewritten as 
  \begin{align}
    R\cdot S=R_{j_1\ldots j_k}S_{j'_1\ldots j'_k}\delta_{j_1j'_1}\ldots\delta_{j_kj'_k}=(R\otimes S)\cdot Z, \label{up1}
  \end{align}
  where the tensor $Z$ is composed by those $\delta$. 
  In the same way, the inner product $(\Phi_1,\Phi_2)$ can be written in the following form, 
  \begin{align*}
    (\Phi_1,\Phi_2)
    =&\int U_1^{r_1}(\mathfrak{p}_1)\otimes U_2^{r_2}(\mathfrak{p}_1)\otimes V_1^{m_1}(\mathfrak{p}_2)\otimes V_2^{m_2}(\mathfrak{p}_2)\cdot Z\md\mathfrak{p}_1\md \mathfrak{p}_2\\
    =& \left(\int U_1^{r_1}(\mathfrak{p}_1)\otimes U_2^{r_2}(\mathfrak{p}_1)\md\mathfrak{p}_1\right)\otimes \left(\int V_1^{m_1}(\mathfrak{p}_2)\otimes V_2^{m_2}(\mathfrak{p}_2)\md \mathfrak{p}_2\right)\cdot Z, 
  \end{align*}
  where $Z$ is some constant tensor. 
  In the case of $U_1^{r_1}\ne U_2^{r_2}$ or $V_1^{m_1}\ne V_2^{r_2}$, at least one of the two integrals is zero, so $(\Phi_1,\Phi_2)=0$. 
  Since the terms in \eqref{orthterms} are linear combinations of the terms in \eqref{terms}, we deduce the orthogonality if the tensor pairs are not identical in two terms. 
  
  Next, we consider the case $U_1^{r_1}=U_2^{r_2}=U^r$ and $V_1^{m_1}=V_2^{m_2}=V^m$. 
  Under this assumption, the different terms in the \eqref{orthterms} must have different $q$. 
  We shall use the following fact: for a symmetric traceless tensor $X$ no less than second order, we have
  $$
  (\mathfrak{i}^qX)\overset{2q+2}{\cdot} \mathfrak{i}^{q+1}=a\mathrm{tr}X=0, 
  $$
  where $a$ is some constant. Therefore, when calculating the inner product $(\Phi_1,\Phi_2)$, we can verify that the integrand will be zero. 
\end{proof}


Finally, we state the approximation result for $\mathscr{M}_2^k$.
\begin{theorem}\label{approx_M2}
  The functions given in \eqref{orthterms} form a complete orthogonal basis in the sense that: if $\mathscr{M}_2^k\in L^2$, then we have
  \begin{align}
    \lim_{n\to\infty}\min_{A\in \mathrm{span}\mathbb{X}^{k}_{n,(-1)^k}}\|\mathscr{M}_2^k-A\|^2=0. 
  \end{align}
  Moreover, assume that $\mathscr{M}_2^k\in H^s$. 
  For $n>k$, the approximation error by the subspace $\mathrm{span}\mathbb{X}^{k}_{n,(-1)^k}$ satisfies 
  \begin{align}
    \min_{A\in \mathrm{span}\mathbb{X}^{k}_{n,(-1)^k}}\|\mathscr{M}_2^k-A\|^2\le C(k)(n-k)^{-s}|\mathscr{M}_2^k|_{H^s}. \label{Hsest}
  \end{align}
  where the constant $C(k)$ only depends on $k$. 
\end{theorem}
\begin{proof}
We only show the inequality \eqref{Hsest}. 
We need to return to the form \eqref{term0}. 
Recall that $\mathscr{M}_2^k\cdot \bm{m}_1^{k_1}(\mathfrak{p}_1)\bm{m}_2^{k_2}(\mathfrak{p}_1)\bm{m}_3^{k_3}(\mathfrak{p}_1)$ is a function of $\mathfrak{p}_1^{-1}\mathfrak{p}_2$. 
Also recall that the projection operator of a scalar function on $SO(3)$, $\pi_N$, is defined in \eqref{projscal}. 
By Proposition \ref{approxbas}, we have
\begin{align*}
&\|\mathscr{M}_2^k\cdot \bm{m}_1^{k_1}(\mathfrak{p}_1)\bm{m}_2^{k_2}(\mathfrak{p}_1)\bm{m}_3^{k_3}(\mathfrak{p}_1)-\pi_{n-k}\big(\mathscr{M}_2^k\cdot \bm{m}_1^{k_1}(\mathfrak{p}_1)\bm{m}_2^{k_2}(\mathfrak{p}_1)\bm{m}_3^{k_3}(\mathfrak{p}_1)\big)\|\\
\le &C(k)(n-k)^s|\mathscr{M}_2^k\cdot \bm{m}_1^{k_1}(\mathfrak{p}_1)\bm{m}_2^{k_2}(\mathfrak{p}_1)\bm{m}_3^{k_3}(\mathfrak{p}_1)|_{H^s}. 
\end{align*}
Let us denote
\begin{align}
  A(\mathfrak{p}_1,\mathfrak{p}_2)=\sum_{k_1+k_2+k_3=k}\bm{m}_1^{k_1}(\mathfrak{p}_1)\bm{m}_2^{k_2}(\mathfrak{p}_1)\bm{m}_3^{k_3}(\mathfrak{p}_1)\frac{k!}{k_1!k_2!k_3!}\pi_{n-k}\big(\mathscr{M}_2^k\cdot \bm{m}_1^{k_1}(\mathfrak{p}_1)\bm{m}_2^{k_2}(\mathfrak{p}_1)\bm{m}_3^{k_3}(\mathfrak{p}_1)\big). \label{approx_component}
\end{align}
Thus, we deduce that
\begin{align}
  \|\mathscr{M}_2^k-A\| \le  C(k)(n-k)^s|\mathscr{M}_2^k|_{H^s}. \label{approx0}
\end{align}
Note that $A(\mathfrak{p}_1,\mathfrak{p}_2)$ can be written in the form 
\begin{align}
  A(\mathfrak{p}_1,\mathfrak{p}_2)=\sum \bm{m}_1^{k_1}(\mathfrak{p}_1)\bm{m}_2^{k_2}(\mathfrak{p}_1)\bm{m}_3^{k_3}(\mathfrak{p}_1)\big(V_1^m(\mathfrak{p}_1)\cdot V^m(\mathfrak{p}_2)\big),\nonumber
\end{align}
where in the summation the symmetric traceless tensors $V_1^m$ and $V^m$ have the same order $m\le n-k$. 
Using Theorem \ref{twoforms}, the above form can be expressed linearly by the terms in \eqref{terms}. 
By the definition of the norm \eqref{norm_tensor} and the property \eqref{switch0} of $\mathscr{M}_2^k$, we obtain another approximation with the same error, 
\begin{align}
  \|\mathscr{M}_2^k(\mathfrak{p}_1,\mathfrak{p}_2)-(-1)^kA(\mathfrak{p}_2,\mathfrak{p}_1)\|^2=&\int \|\mathscr{M}_2^k(\mathfrak{p}_1,\mathfrak{p}_2)-(-1)^kA(\mathfrak{p}_2,\mathfrak{p}_1)\|_F^2\md\mathfrak{p}_1\md \mathfrak{p}_2\nonumber\\
  =& \int \|(-1)^k\mathscr{M}_2^k(\mathfrak{p}_2,\mathfrak{p}_1)-(-1)^kA(\mathfrak{p}_2,\mathfrak{p}_1)\|_F^2\md\mathfrak{p}_1\md \mathfrak{p}_2\nonumber\\
  =& \int \|(-1)^k\mathscr{M}_2^k(\mathfrak{p}_1,\mathfrak{p}_2)-(-1)^kA(\mathfrak{p}_1,\mathfrak{p}_2)\|_F^2\md\mathfrak{p}_2\md \mathfrak{p}_1\nonumber\\
  =&\|\mathscr{M}_2^k(\mathfrak{p}_1,\mathfrak{p}_2)-A(\mathfrak{p}_1,\mathfrak{p}_2)\|^2. \nonumber
\end{align}
We notice that \eqref{terms} requires $r\le m+k\le n$. 
Therefore, $A(\mathfrak{p}_1,\mathfrak{p}_2)+(-1)^kA(\mathfrak{p}_2,\mathfrak{p}_1)$ belongs to $\mathrm{span}\mathbb{X}^{k}_{n,(-1)^k}$. 
So we arrive at 
\begin{align}
  &\|\mathscr{M}_2^k(\mathfrak{p}_1,\mathfrak{p}_2)-\frac{1}{2}\big(A(\mathfrak{p}_1,\mathfrak{p}_2)+(-1)^kA(\mathfrak{p}_2,\mathfrak{p}_1)\big)\|\nonumber\\
  \le &\frac{1}{2}\Big(\|\mathscr{M}_2^k(\mathfrak{p}_1,\mathfrak{p}_2)-A(\mathfrak{p}_1,\mathfrak{p}_2)\|+\|\mathscr{M}_2^k(\mathfrak{p}_1,\mathfrak{p}_2)-(-1)^kA(\mathfrak{p}_2,\mathfrak{p}_1)\|\Big)\nonumber\\
  =&\|\mathscr{M}_2^k(\mathfrak{p}_1,\mathfrak{p}_2)-A(\mathfrak{p}_1,\mathfrak{p}_2)\|. 
\end{align}
Together with \eqref{approx0}, we get the result. 
\end{proof}



\subsection{Clusters of three or more molecules}
We turn to the interaction kernels for clusters of three or more molecules, 
for which the same procedure of dealing with $\mathscr{M}_2^k$ is carried out. 
The invariance when the whole cluster is displaced or rotated requires 
\begin{align*}
  &\mathscr{G}_l(\mathfrak{t}\bm{r}_2,\ldots,\mathfrak{t}\bm{r}_l,\mathfrak{tp}_1,\ldots,\mathfrak{tp}_l)=\mathscr{G}_l(\bm{r}_2,\ldots,\bm{r}_l,\mathfrak{p}_1,\ldots,\mathfrak{p}_l),\quad \mathfrak{t}\in SO(3), 
\end{align*}
yielding $\mathscr{M}_l^{k_2,\ldots,k_l}(\mathfrak{tp}_1,\ldots,\mathfrak{tp}_l)=\mathscr{M}_l^{k_2,\ldots,k_l}(\mathfrak{p}_1,\ldots,\mathfrak{p}_l)$.
Therefore, $\mathscr{M}_l^{k_2,\ldots,k_l}$ are functions of $\mathfrak{p}_1^{-1}\mathfrak{p}_j$ for $j=2,\ldots,l$. 
We could then expand $\mathscr{M}_l^{k_2,\ldots,k_l}$ about these variables like in \eqref{term0}, and decompose the tensor $Y(\mathfrak{p}_1)$ into a symmetric traceless tensor. 
As a result, we obtain some terms given by multi-linear maps from $l$ symmetric traceless tensors to another tensor (cf. \eqref{up1}), 
\begin{align}
  \varPhi\big(U_1(\mathfrak{p}_1),\ldots,U_l(\mathfrak{p}_l)\big)=\big(U_1(\mathfrak{p}_1)\otimes\ldots\otimes U_l(\mathfrak{p}_l)\big)_{i_1\ldots i_s}Z_{i_{\tau_1}\ldots i_{\tau_w}j_1\ldots j_t},\nonumber\\
  s-w+t=k_2+\ldots k_l, \label{UZ}
\end{align}
where $Z$ is a tensor containing some $\delta$ and $\epsilon$. 
The terms in \eqref{terms}, giving bilinear maps about symmetric traceless tensors $U^r(\mathfrak{p}_1)$ and $V^m(\mathfrak{p}_2)$, are actually a special case of \eqref{UZ}.
Following the arguments in the proof of Theorem \ref{orth_M2}, we immediately obtain the orthogonality if the tensors are not identical in $\varPhi$. 
\begin{theorem}\label{expansion_tensor}
  $\varPhi\big(U_1(\mathfrak{p}_1),\ldots,U_l(\mathfrak{p}_l)\big)$ and $\varPhi'\big(U'_1(\mathfrak{p}_1),\ldots,U'_l(\mathfrak{p}_l)\big)$ are orthogonal if $U_i\cdot U'_i=0$ for some $i$. 
\end{theorem}
It is also straightforward to state the approximation results like Theorem \ref{approx_M2} when requiring $U_i$ to be no greater than $n$-th order, which we would omit here. 

The difficulty in expanding $\mathscr{M}_l^{k_2,\ldots,k_l}$ is to identify linearly independent terms.
For the same $(U_1,\ldots,U_l)$ there are multiple terms that might have complicated linear relations, especially when combined with the arguments of switching labels (cf. \eqref{switch00}). 
We shall discuss two cases, $\mathscr{M}_3^{0,0}$ and $\mathscr{M}_4^{0,0,0}$, which are expected to be important in applications. 

\subsubsection{Expansion of $\mathscr{M}_3^{0,0}$}
To prepare for our discussion, we introduce the notation $\mathpzc{a}_3$ for a scalar by contracting indices of three symmetric traceless tensors $(U_1^{n_1},U_2^{n_2},U_3^{n_3})$, 
\begin{subequations}\label{terms_a3}
\begin{align}
  &\mathpzc{a}_3(U_1^{n_1},U_2^{n_2},U_3^{n_3};l_{12},l_{13},l_{23})\nonumber\\
  =&(U_1^{n_1})_{i^{(12)}_1\ldots i^{(12)}_{l_{12}}i^{(13)}_1\ldots i^{(13)}_{l_{13}}}(U_2^{n_2})_{i^{(12)}_1\ldots i^{(12)}_{l_{12}}i^{(23)}_1\ldots i^{(23)}_{l_{23}}}(U_3^{n_3})_{i^{(13)}_1\ldots i^{(13)}_{l_{13}}i^{(23)}_1\ldots i^{(23)}_{l_{23}}}, \label{noeps}\\
  &\mathpzc{a}_3(U_1^{n_1},U_2^{n_2},U_3^{n_3};l_{12},l_{13},l_{23},(123))\nonumber\\
  =&\epsilon_{j_1j_2j_3}(U_1^{n_1})_{j_1i^{(12)}_1\ldots i^{(12)}_{l_{12}}i^{(13)}_1\ldots i^{(13)}_{l_{13}}}(U_2^{n_2})_{j_2i^{(12)}_1\ldots i^{(12)}_{l_{12}}i^{(23)}_1\ldots i^{(23)}_{l_{23}}}(U_3^{n_3})_{j_3i^{(13)}_1\ldots i^{(13)}_{l_{13}}i^{(23)}_1\ldots i^{(23)}_{l_{23}}}. \label{eps}
\end{align}
\end{subequations}
The nonnegative integers $l_{ij}$ represent the number of indices contracted between $U_i^{n_i}$ and $U_j^{n_j}$, and $(\tau_1\tau_2\tau_3)=(123)$ means that there is an $\epsilon_{j_1j_2j_3}$ such that $j_1$ appears in $U_{\tau_1}^{n_{\tau_1}}=U_1^{n_1}$, $j_2$ appears in $U_2^{n_2}$, and $j_3$ appears in $U_3^{n_3}$. 
The parameters $l_{ij}$ here are actually redundant. Actually, we have 
\begin{align}
  l_{12}+l_{13}=n_1,\  l_{12}+l_{23}=n_2,\  l_{13}+l_{23}=n_3, \label{tn3_1}
\end{align}
in \eqref{noeps}, where we require that $K=n_1+n_2+n_3$ is even and $K\ge 2n_i$ for $i=1,2,3$. 
Similarly, we have 
\begin{align}
  l_{12}+l_{13}=n_1-1,\  l_{12}+l_{23}=n_2-1,\  l_{13}+l_{23}=n_3-1, \label{tn3_2}
\end{align}
in \eqref{eps}, where we require that $n_i\ge 1$, $K=n_1+n_2+n_3$ is odd, and $K\ge 2n_i+1$. 
However, we still keep $l_{ij}$ in the expression, because we will use similar notations for four tensors. 
It is noticed that when permutating the three tensors in $\mathpzc{a}_3$, we could get some identical or opposite terms, such as
\begin{subequations}\label{a3swlb}
\begin{align}
  \mathpzc{a}_3(U_2^{n_2},U_1^{n_1},U_3^{n_3};l_{12},l_{23},l_{13}) 
  =&\mathpzc{a}_3(U_1^{n_1},U_2^{n_2},U_3^{n_3};l_{12},l_{13},l_{23}), \\
  \mathpzc{a}_3(U_2^{n_2},U_1^{n_1},U_3^{n_3};l_{12},l_{23},l_{13},(123))
  =&-\mathpzc{a}_3(U_1^{n_1},U_2^{n_2},U_3^{n_3};l_{12},l_{13},l_{23},(123)). 
\end{align}
\end{subequations}
Thus, once the three tensors are chosen, we can fix how they are arranged in $\mathpzc{a}_3$. 
In particular, if two tensors are identical in \eqref{eps}, the term equals to zero. 

Now we are ready to expand $\mathscr{M}_3^{0,0}$. As we have mentioned, it is a function of $\mathfrak{p}_1^{-1}\mathfrak{p}_2$ and $\mathfrak{p}_1^{-1}\mathfrak{p}_3$.
When expanding about these two variables, the resulting terms can be written as 
\begin{align}
  \big(Y_2^{n_2}(\mathfrak{p}_1)\cdot U_2^{n_2}(\mathfrak{p}_2)\big)\big(Y_3^{n_3}(\mathfrak{p}_1)\cdot U_3^{n_3}(\mathfrak{p}_3)\big)
  =\big(Y_2^{n_2}(\mathfrak{p}_1)\otimes Y_3^{n_3}(\mathfrak{p}_1)\big)\cdot \big(U_2^{n_2}(\mathfrak{p}_2)\otimes U_3^{n_3}(\mathfrak{p}_3)\big), \label{term0_a3}
\end{align}
where $Y_i^{n_i}$ and $U_i^{n_i}$ are symmetric traceless tensors.
After we decompose $Y_2^{n_2}\otimes Y_3^{n_3}$ into symmetric traceless tensors, the above terms can be expressed linearly by terms in \eqref{UZ} where $l=3$, $s=w$ and $t=0$. 
Using the notation $\mathpzc{a}_3$, they are given by 
\begin{subequations}\label{terms_a3p}
\begin{align}
  &\mathpzc{a}_3\big(U_1^{n_1}(\mathfrak{p}_{1}),U_2^{n_2}(\mathfrak{p}_{2}),U_3^{n_3}(\mathfrak{p}_{3});l_{12},l_{13},l_{23}\big),\label{a3p_noeps}\\
  &\mathpzc{a}_3\big(U_1^{n_1}(\mathfrak{p}_{1}),U_2^{n_2}(\mathfrak{p}_{2}),U_3^{n_3}(\mathfrak{p}_{3});l_{12},l_{13},l_{23},(123)\big). \label{a3p_eps}
\end{align}
\end{subequations}
Similar to Theorem \ref{twoforms}, let us fix $U_2^{n_2}$ and $U_3^{n_3}$, and examine the linearly independent terms, by comparing \eqref{term0_a3} and \eqref{terms_a3p}.
The former can be expressed linearly by the latter. 
On the other hand, because $l_{ij}\ge 0$ in \eqref{tn3_1} and \eqref{tn3_2}, the tensor order $n_1$ in \eqref{terms_a3p} ranges from $|n_2-n_3|$ to $n_2+n_3$.
Once $n_1$ is determined, the $l_{ij}$ are also determined, and $U_1^{n_1}$ has $2n_1+1$ choices. 
Hence, the total number of choices of $U_1^{n_1}$ is 
\begin{align*}
  \sum_{r=|n_2-n_3|}^{n_2+n_3}(2r+1)=(2n_2+1)(2n_3+1), 
\end{align*}
which is equal to the dimension of $Y_2^{n_2}\otimes Y_3^{n_3}$. 
Therefore, when we let $U_i^{n_i}$ be the tensors in $\mathbb{W}^{n_i}$, the terms given by \eqref{terms_a3p} are linearly independent. 

Then, we take the switching of labels into consideration.
It requires $\mathscr{M}_3^{0,0}(\mathfrak{p}_{\sigma(1)},\mathfrak{p}_{\sigma(2)},\mathfrak{p}_{\sigma(3)})=\mathscr{M}_3^{0,0}(\mathfrak{p}_1,\mathfrak{p}_2,\mathfrak{p}_3)$ for any permutation $\sigma$ of $1,2,3$. 
Thus, the expansion can only have the terms below, 
\begin{subequations}\label{a3perm}
\begin{align}
  &\sum_{\sigma}\mathpzc{a}_3(U_1^{n_1}(\mathfrak{p}_{\sigma(1)}),U_2^{n_2}(\mathfrak{p}_{\sigma(2)}),U_3^{n_3}(\mathfrak{p}_{\sigma(3)});l_{12},l_{13},l_{23}), \label{a3perm1}\\
  &\sum_{\sigma}\mathpzc{a}_3(U_1^{n_1}(\mathfrak{p}_{\sigma(1)}),U_2^{n_2}(\mathfrak{p}_{\sigma(2)}),U_3^{n_3}(\mathfrak{p}_{\sigma(3)});l_{12},l_{13},l_{23},(123)), \label{a3perm2}
\end{align}
\end{subequations}
where the $\sigma$ in the summation takes all the permutations. 
As we mentioned in \eqref{a3swlb}, the terms are invariant or become opposite when interchanging the three tensors $U_i^{n_i}$. 
In particular, in \eqref{a3perm2}, if any two of $U_i^{n_i}$ are identical, then the term vanishes. 


\subsubsection{Expansion of $\mathscr{M}_4^{0,0,0}$}

We expand $\mathscr{M}_4^{0,0,0}$ about three variables $\mathfrak{p}_1^{-1}\mathfrak{p}_j$ for $j=2,3,4$, to obtain the terms 
\begin{align}
  &\big(Y_2^{n_2}(\mathfrak{p}_1)\cdot U_2^{n_2}(\mathfrak{p}_2)\big)\big(Y_3^{n_3}(\mathfrak{p}_1)\cdot U_3^{n_3}(\mathfrak{p}_3)\big)\big(Y_4^{n_4}(\mathfrak{p}_1)\cdot U_4^{n_4}(\mathfrak{p}_4)\big)\nonumber\\
  &\qquad =\big(Y_2^{n_2}(\mathfrak{p}_1)\otimes Y_3^{n_3}(\mathfrak{p}_1)\otimes Y_4^{n_4}(\mathfrak{p}_1)\big)\cdot \big(U_2^{n_2}(\mathfrak{p}_2)\otimes U_3^{n_3}(\mathfrak{p}_3)\otimes U_4^{n_4}(\mathfrak{p}_4)\big). \label{M4expansion}
\end{align}
The decomposition of $Y_2^{n_2}\otimes Y_3^{n_3}\otimes Y_4^{n_4}$ into symmetric traceless tensors is followed. 
Similar to \eqref{terms_a3}, we use the notation $\mathpzc{a}_4$ for a scalar generated by contraction of four tensors,
\begin{subequations}\label{a4_def}
\begin{align}
  &\mathpzc{a}_4\Big(U_i^{n_i}\big|_{i=1}^4;l_{12},l_{13},l_{14},l_{23},l_{24},l_{34}\Big),\\
  &\mathpzc{a}_4\Big(U_i^{n_i}\big|_{i=1}^4;l_{12},l_{13},l_{14},l_{23},l_{24},l_{34},(\tau_1\tau_2\tau_3)\Big), 
\end{align}
\end{subequations}
where the integers $l_{ij}$ represent how many indices are contracted between $U_i^{n_i}$ and $U_j^{n_j}$; $(\tau_1\tau_2\tau_3)$ means that there is an $\epsilon$ to contract the indices in the way $\epsilon_{j_1j_2j_3}(U_{\tau_1}^{n_{\tau_1}})_{j_1\ldots}(U_{\tau_2}^{n_{\tau_2}})_{j_2\ldots}(U_{\tau_3}^{n_{\tau_3}})_{j_3\ldots}$.
These nonnegative integers shall satisfy 
\begin{align}
  \sum_{j=i+1}^4l_{ij}+\sum_{j=1}^{i-1}l_{ji}=n_i-b_i,\quad b_i=\left\{
  \begin{array}{ll}
    1,& i=\tau_1,\tau_2,\text{ or }\tau_3,\\
    0,& \text{otherwise}. 
  \end{array}
  \right.\label{lrelation}
\end{align}
As we explained in \eqref{a3swlb}, when permutating the tensors $U_i^{n_i}$, some terms are identical or opposite. 

Together with the symmetry of switching the labels, we eventually obtain the terms
\begin{subequations}\label{terms_a4}
\begin{align}
  &\sum_{\sigma}\mathpzc{a}_4\Big(U_i^{n_i}(\mathfrak{p}_{\sigma(i)})\big|_{i=1}^4;l_{ij}\big|_{1\le i<j\le 4}\Big),\label{terms_a4_1}\\
  &\sum_{\sigma}\mathpzc{a}_4\Big(U_i^{n_i}(\mathfrak{p}_{\sigma(i)})\big|_{i=1}^4;l_{ij}\big|_{1\le i<j\le 4},(\tau_1\tau_2\tau_3)\Big),\label{terms_a4_2}
\end{align}
\end{subequations}
where $l_{ij}$ satisfy \eqref{lrelation}. 
However, unlike the cases we discussed above, the terms in \eqref{terms_a4} still have linear relations. 
Below, we write down the linearly independent terms. 
Denote $K=n_1+n_2+n_3+n_4$ and $D=n_1+n_2-n_3-n_4$.
Note that in \eqref{terms_a4_1} $K$ and $D$ are even with $K\ge 2n_i$, while in \eqref{terms_a4_2} $K$ and $D$ are odd and $K\ge 2n_i+1$.
\begin{enumerate}
\item The four tensors $U_i^{n_i}$ are mutually unequal. 
  \begin{itemize}
  \item For \eqref{terms_a4_1}, $D$ is even. If $D\le 0$, we require $l_{12}\le 1$; if $D\ge 0$ we require $l_{34}\le 1$. Notice that when $D=0$, by \eqref{lrelation} we have $n_1+n_2-n_3-n_4=2l_{12}-2l_{34}=0$.
  \item For \eqref{terms_a4_2}, $D$ is odd. If $D\le -1$, we let $(\tau_1\tau_2\tau_3)=(134),(234)$ and $l_{12}=0$; if $D\ge 1$, we let $(\tau_1\tau_2\tau_3)=(123),(124)$ and $l_{34}=0$. 
  \end{itemize}
\item Two tensors are equal, but they are not equal to the other two. We place these two tensors in the first two, i.e. $U_1^{n_1}=U_2^{n_2}$. 
  \begin{itemize}
  \item For \eqref{terms_a4_1}, if $D\le 0$, we require $l_{12}\le 1$ and $l_{13}\le l_{23}$; if $D\ge 0$ we require $l_{34}\le 1$ and $l_{13}\le l_{23}$. 
  \item For \eqref{terms_a4_2}, if $D\le -1$, we let $(\tau_1\tau_2\tau_3)=(134)$ and $l_{12}=0$; if $D\ge 1$, we let $(\tau_1\tau_2\tau_3)=(123),(124)$ and $l_{34}=0$, $l_{13}<l_{23}$. 
  \end{itemize}
  For the case $U_1^{n_1}=U_2^{n_2}$ and $U_3^{n_3}=U_4^{n_4}$, only \eqref{terms_a4_1} appears since $D$ is even. The above conditions still apply. 
\item Three tensors are equal. We let $U_1^{n_1}=U_2^{n_2}=U_3^{n_3}$. 
  \begin{itemize}
  \item For \eqref{terms_a4_1}, we require $l_{12}=l_{13}\le l_{23}$. 
  \item For \eqref{terms_a4_2}, let $(\tau_1\tau_2\tau_3)=(124)$ and we look at $D=n_1-n_4$. If $D\le -1$, we require $l_{12}=l_{13}< l_{23}$; if $D\ge 1$, we require $l_{34}=l_{24}<l_{14}$. 
  \end{itemize}
  If four tensors are equal, only \eqref{terms_a4_1} appears and the above conditions still apply. 
\end{enumerate}
The derivation is tedious and is left to Appendix. 
Here, we explain the conditions stated above by a couple of examples.
We consider the case 3 with $n_1=n_2=n_3=3$, and discuss two cases: $n_4=3$ and $n_4=1$.
\begin{itemize}
\item $n_4=3$. From \eqref{lrelation}, we derive that $2l_{34}-2l_{12}=n_3+n_4-n_1-n_2=0$. So we have $l_{12}=l_{34}$, $l_{13}=l_{24}$, $l_{14}=l_{23}$. We also deduce from \eqref{lrelation} that 
  \begin{align*}
    2(l_{12}+l_{13}+l_{23}+l_{14}+l_{24}+l_{34})=n_1+n_2+n_3+n_4. 
  \end{align*}
  It implies that $l_{12}+l_{13}+l_{23}=3$. Therefore, with the condition $l_{12}=l_{13}\le l_{23}$, we find two choices $(l_{12},l_{13},l_{23})=(0,0,3),\,(1,1,1)$.
\item $n_4=1$. Similarly, we can derive that $l_{12}-l_{34}=l_{13}-l_{24}=l_{14}-l_{23}=1$. Since $l_{ij}\ge 0$, we need $l_{12},l_{13},l_{23}\ge 1$. We can also find that $l_{12}+l_{13}+l_{23}=4$, which only gives us one choice $(l_{12},l_{13},l_{23})=(1,1,2)$. 
\end{itemize}

\subsection{Summary of explicit expressions}
Here, we summarize the explicit formulae for $\mathscr{M}_2^k$ for $0\le k\le 4$, $\mathscr{M}_3^{0,0}$ and $\mathscr{M}_4^{0,0,0}$ in Table \ref{sum_term}. 
When taking these terms back into \eqref{TlExp}, the integrals $\md\mathfrak{p}_i$ are decoupled like what is done in \eqref{sepvar0}, leading to the terms in the free energy expressed by tensors that are also listed in Table \ref{sum_term}. 
We do not distinguish terms that coincide under integration by parts. 
For example, we consider 
\begin{align*}
&\int \big(U^{n-1}(\mathfrak{p}_1)\overset{n-1}{\cdot}V^n(\mathfrak{p}_2)-V^n(\mathfrak{p}_1)\overset{n-1}{\cdot}U^{n-1}(\mathfrak{p}_2)\big)f(\bm{x},\mathfrak{p}_1)f(\bm{x},\mathfrak{p}_2)\, \md\mathfrak{p}_1\md\mathfrak{p}_2\\
=&\langle \tsn{n-1}{U}\rangle_{i_1\ldots i_{n-1}}\partial_j\langle \tsn{n}{V}\rangle_{i_1\ldots i_{n-1}j}-\langle \tsn{n}{V}\rangle_{i_1\ldots i_{n-1}j}\partial_j\langle \tsn{n-1}{U}\rangle_{i_1\ldots i_{n-1}}\\ 
=&2\langle \tsn{n-1}{U}\rangle_{i_1\ldots i_{n-1}}\partial_j\langle \tsn{n}{V}\rangle_{i_1\ldots i_{n-1}j}-\partial_j\left(\langle \tsn{n}{V}\rangle_{i_1\ldots i_{n-1}j}\langle \tsn{n-1}{U}\rangle_{i_1\ldots i_{n-1}}\right). 
\end{align*}
When integrated about $\md\bm{x}$, the second term in the last line leads to a surface integral. In this sense, we regard $\langle \tsn{n-1}{U}\rangle_{i_1\ldots i_{n-1}}\partial_j\langle \tsn{n}{V}\rangle_{i_1\ldots i_{n-1}j}$ and $\langle \tsn{n}{V}\rangle_{i_1\ldots i_{n-1}j}\partial_j\langle \tsn{n-1}{U}\rangle_{i_1\ldots i_{n-1}}$ as the same term. 

\begin{table}
  \centering
  \caption{Linearly independent terms in the expansion and the corresponding terms in the free energy. All the tensors are symmetric traceless. The notation $\langle U\rangle$ represents the average of $U(\mathfrak{p})$ about the density $f(\bm{x},\mathfrak{p})$. \label{sum_term}}\small
  \begin{tabular}{c|c|c}\hline
  & Orientational expansion & Free energy\\\hline
  $\mathscr{M}_2^0$ & $U^n(\mathfrak{p}_1){\cdot}V^n(\mathfrak{p}_2)+V^n(\mathfrak{p}_1){\cdot}U^n(\mathfrak{p}_2)$ & 
  $\langle\tsn{n}{U}\rangle_{i_1\ldots i_n}\langle\tsn{n}{V}\rangle_{i_1\ldots i_n}=\tsn{n}{U}\cdot \tsn{n}{V}$ \\\hline
  $\mathscr{M}_2^1$ & $U^{n-1}(\mathfrak{p}_1)\overset{n-1}{\cdot}V^n(\mathfrak{p}_2)-V^n(\mathfrak{p}_1)\overset{n-1}{\cdot}U^{n-1}(\mathfrak{p}_2)$ & 
  $\langle \tsn{n-1}{U}\rangle_{i_1\ldots i_{n-1}}\partial_j\langle \tsn{n}{V}\rangle_{i_1\ldots i_{n-1}j}$\\
  & $U^n(\mathfrak{p}_1)\overset{n-1}{\times} V^n(\mathfrak{p}_2)+V^n(\mathfrak{p}_1)\overset{n-1}{\times} U^n(\mathfrak{p}_2)$ &
  $\epsilon_{ijk}\langle U^{n}\rangle_{i_1\ldots i_{n-1}i}\partial_k\langle V^{n}\rangle_{i_1\ldots i_{n-1}j}$\\\hline
  $\mathscr{M}_2^2$ & $\mathfrak{i}(U^n(\mathfrak{p}_1){\cdot}V^n(\mathfrak{p}_2)+V^n(\mathfrak{p}_1){\cdot}U^n(\mathfrak{p}_2))$ & 
  $\partial_j\langle U^n\rangle_{i_1\ldots i_n}\partial_j\langle V^n\rangle_{i_1\ldots i_n}$\\
  & $U^n(\mathfrak{p}_1)\overset{n-1}{\cdot}V^n(\mathfrak{p}_2)+V^n(\mathfrak{p}_1)\overset{n-1}{\cdot}U^n(\mathfrak{p}_2)$ & 
  $\partial_{j_1}\langle U^n\rangle_{i_1\ldots i_{n-1}j_1}\partial_{j_2}\langle V^n\rangle_{i_1\ldots i_{n-1}j_2}$ \\
  & $U^{n-2}(\mathfrak{p}_1)\overset{n-2}{\cdot}V^n(\mathfrak{p}_2)+V^n(\mathfrak{p}_1)\overset{n-2}{\cdot}U^{n-2}(\mathfrak{p}_2)$ & 
  $\partial_j\langle U^{n-2}\rangle_{i_1\ldots i_{n-2}}\partial_k\langle V^n\rangle_{i_1\ldots i_{n-2}jk}$\\
  & $U^n(\mathfrak{p}_1)\overset{n-2}{\times} V^{n-1}(\mathfrak{p}_2)-V^{n-1}(\mathfrak{p}_1)\overset{n-2}{\times} U^n(\mathfrak{p}_2)$ & 
  $\epsilon_{ijk}\partial_l\langle U^{n}\rangle_{i_1\ldots i_{n-2}il}\partial_k\langle V^{n-1}\rangle_{i_1\ldots i_{n-2}j}$\\\hline
  $\mathscr{M}_2^3$ & $\mathfrak{i}U^{n-1}(\mathfrak{p}_1)\overset{n-1}{\cdot}V^n(\mathfrak{p}_2)-\mathfrak{i}V^n(\mathfrak{p}_1)\overset{n-1}{\cdot}U^{n-1}(\mathfrak{p}_2)$ & 
  $\partial_{j_1}\langle \tsn{n-1}{U}\rangle_{i_1\ldots i_{n-1}}\partial_{j_1j_2}\langle \tsn{n}{V}\rangle_{i_1\ldots i_{n-1}j_2}$\\
  & $\mathfrak{i}U^n(\mathfrak{p}_1)\overset{n-1}{\times} V^n(\mathfrak{p}_2)+\mathfrak{i}V^n(\mathfrak{p}_1)\overset{n-1}{\times} U^n(\mathfrak{p}_2)$ & 
  $\epsilon_{ijk}\partial_{l}\langle U^{n}\rangle_{i_1\ldots i_{n-1}i}\partial_{kl}\langle V^{n}\rangle_{i_1\ldots i_{n-1}j}$\\
  & $U^{n-1}(\mathfrak{p}_1)\overset{n-2}{\cdot}V^n(\mathfrak{p}_2)-V^n(\mathfrak{p}_1)\overset{n-2}{\cdot}U^{n-1}(\mathfrak{p}_2)$ & 
  $\partial_{j_1}\langle \tsn{n-1}{U}\rangle_{i_1\ldots i_{n-2}j_1}\partial_{j_2j_3}\langle \tsn{n}{V}\rangle_{i_1\ldots i_{n-2}j_2j_3}$\\
  & $U^n(\mathfrak{p}_1)\overset{n-2}{\times} V^n(\mathfrak{p}_2)+V^n(\mathfrak{p}_1)\overset{n-2}{\times} U^n(\mathfrak{p}_2)$ & 
  $\epsilon_{ijk}\partial_{j_1}\langle U^{n}\rangle_{i_1\ldots i_{n-2}j_1i}\partial_{kj_2}\langle V^{n}\rangle_{i_1\ldots i_{n-2}j_2j}$\\
  & $U^{n-3}(\mathfrak{p}_1)\overset{n-3}{\cdot}V^n(\mathfrak{p}_2)-V^n(\mathfrak{p}_1)\overset{n-3}{\cdot}U^{n-3}(\mathfrak{p}_2)$ & 
  $\partial_{j_1}\langle \tsn{n-3}{U}\rangle_{i_1\ldots i_{n-3}}\partial_{j_2j_3}\langle \tsn{n}{V}\rangle_{i_1\ldots i_{n-2}j_1j_2j_3}$\\
  & $U^{n-2}(\mathfrak{p}_1)\overset{n-3}{\times} V^n(\mathfrak{p}_2)+V^n(\mathfrak{p}_1)\overset{n-3}{\times} U^{n-2}(\mathfrak{p}_2)$ & 
  $\epsilon_{ijk}\partial_{k}\langle U^{n-2}\rangle_{i_1\ldots i_{n-3}i}\partial_{j_1j_2}\langle V^{n}\rangle_{i_1\ldots i_{n-3}j_1j_2j}$\\\hline
  $\mathscr{M}_2^4$ & $\mathfrak{i}^2(U^n(\mathfrak{p}_1){\cdot}V^n(\mathfrak{p}_2)+V^n(\mathfrak{p}_1){\cdot}U^n(\mathfrak{p}_2))$ & 
  $\partial_{j_1j_2}\langle U^n\rangle_{i_1\ldots i_n}\partial_{j_1j_2}\langle V^n\rangle_{i_1\ldots i_n}$\\
  & $\mathfrak{i}U^n(\mathfrak{p}_1)\overset{n-1}{\cdot}V^n(\mathfrak{p}_2)+\mathfrak{i}V^n(\mathfrak{p}_1)\overset{n-1}{\cdot}U^n(\mathfrak{p}_2)$ & 
  $\partial_{j_1j_3}\langle U^n\rangle_{i_1\ldots i_{n-1}j_1}\partial_{j_2j_3}\langle V^n\rangle_{i_1\ldots i_{n-1}j_2}$ \\
  & $U^n(\mathfrak{p}_1)\overset{n-2}{\cdot}V^n(\mathfrak{p}_2)+V^n(\mathfrak{p}_1)\overset{n-2}{\cdot}U^n(\mathfrak{p}_2)$ & 
  $\partial_{j_1j_2}\langle U^n\rangle_{i_1\ldots i_{n-2}j_1j_2}\partial_{j_3j_4}\langle V^n\rangle_{i_1\ldots i_{n-2}j_3j_4}$ \\
  & $\mathfrak{i}U^{n-2}(\mathfrak{p}_1)\overset{n-2}{\cdot}V^n(\mathfrak{p}_2)+\mathfrak{i}V^n(\mathfrak{p}_1)\overset{n-2}{\cdot}U^{n-2}(\mathfrak{p}_2)$ & 
  $\partial_{j_1j_2}\langle U^{n-2}\rangle_{i_1\ldots i_{n-2}}\partial_{j_1j_3}\langle V^n\rangle_{i_1\ldots i_{n-2}j_2j_3}$\\
  & $U^{n-2}(\mathfrak{p}_1)\overset{n-3}{\cdot}V^n(\mathfrak{p}_2)+V^n(\mathfrak{p}_1)\overset{n-3}{\cdot}U^{n-2}(\mathfrak{p}_2)$ & 
  $\partial_{j_1j_2}\langle U^{n-2}\rangle_{i_1\ldots i_{n-3}j_1}\partial_{j_3j_4}\langle V^n\rangle_{i_1\ldots i_{n-3}j_2j_3j_4}$\\
  & $\mathfrak{i}U^{n-1}(\mathfrak{p}_1)\overset{n-2}{\times} V^n(\mathfrak{p}_2)-\mathfrak{i}V^n(\mathfrak{p}_1)\overset{n-2}{\times} U^{n-1}(\mathfrak{p}_2)$ & 
  $\epsilon_{ijk}\partial_{j_1j_2}\langle U^{n}\rangle_{i_1\ldots i_{n-2}j_2i}\partial_{kj_1}\langle V^{n-1}\rangle_{i_1\ldots i_{n-2}j}$\\
  & $U^n(\mathfrak{p}_1)\overset{n-3}{\times} V^{n-1}(\mathfrak{p}_2)-V^{n-1}(\mathfrak{p}_1)\overset{n-3}{\times} U^n(\mathfrak{p}_2)$ & 
  $\epsilon_{ijk}\partial_{j_1j_2}\langle U^{n}\rangle_{i_1\ldots i_{n-3}j_1j_2i}\partial_{kj_3}\langle V^{n-1}\rangle_{i_1\ldots i_{n-3}j_3j}$\\
  & $U^{n-3}(\mathfrak{p}_1)\overset{n-4}{\times} V^n(\mathfrak{p}_2)-V^n(\mathfrak{p}_1)\overset{n-4}{\times} U^{n-3}(\mathfrak{p}_2)$ & 
  $\epsilon_{ijk}\partial_{kj_1}\langle U^{n-3}\rangle_{i_1\ldots i_{n-4}i}\partial_{j_2j_3}\langle V^{n}\rangle_{i_1\ldots i_{n-4}j_1j_2j_3j}$\\
  & $U^{n-4}(\mathfrak{p}_1)\overset{n-4}{\cdot}V^n(\mathfrak{p}_2)+V^n(\mathfrak{p}_1)\overset{n-4}{\cdot}U^{n-4}(\mathfrak{p}_2)$ & 
  $\partial_{j_1j_2}\langle U^{n-4}\rangle_{i_1\ldots i_{n-4}}\partial_{j_3j_4}\langle V^n\rangle_{i_1\ldots i_{n-4}j_1j_2j_3j_4}$\\\hline
  $\mathscr{M}_3^{0,0}$ 
  & $\sum_{\sigma}\mathpzc{a}_3\Big(U_i^{n_i}(\mathfrak{p}_{\sigma(i)})\big|_{i=1}^3;l_{ij}\big|_{1\le i<j\le 3}\Big)$ & $\mathpzc{a}_3\Big(\langle U_i^{n_i}\rangle\big|_{i=1}^3;l_{ij}\big|_{1\le i<j\le 3}\Big)$\\\cline{2-3}
  & \multicolumn{2}{l}{$K=n_1+n_2+n_3$ even, $K\ge 2n_i$}\\\cline{2-3}
  & $\sum_{\sigma}\mathpzc{a}_3\Big(U_i^{n_i}(\mathfrak{p}_{\sigma(i)})\big|_{i=1}^3;l_{ij}\big|_{1\le i<j\le 3},(123)\Big)$ & $\mathpzc{a}_3\Big(\langle U_i^{n_i}\rangle\big|_{i=1}^3;l_{ij}\big|_{1\le i<j\le 3},(123)\Big)$\\\cline{2-3}
  & \multicolumn{2}{l}{$K=n_1+n_2+n_3$ odd, $K-1\ge 2n_i$; $U_i^{n_i}$ mutually unequal}\\\hline
  $\mathscr{M}_4^{0,0,0}$ & $\sum_{\sigma}\mathpzc{a}_4\Big(U_i^{n_i}(\mathfrak{p}_{\sigma(i)})\big|_{i=1}^4;l_{ij}\big|_{1\le i<j\le 4}\Big)$ & $\mathpzc{a}_4\Big(\langle U_i^{n_i}\rangle\big|_{i=1}^4;l_{ij}\big|_{1\le i<j\le 4}\Big)$\\\cline{2-3}
  & \multicolumn{2}{l}{$K=n_1+n_2+n_3+n_4$ even, $K\ge 2n_i$, $D=n_1+n_2-n_3-n_4=2l_{12}-2l_{34}$}\\
  & \multicolumn{2}{l}{$U_i^{n_i}$ mutually unequal: If $D\le 0$, then $l_{12}\le 1$; if $D\ge 0$, then $l_{34}\le 1$}\\
  & \multicolumn{2}{l}{$U_1^{n_1}=U_2^{n_2}$: If $D\le 0$, then $l_{12}\le 1$, $l_{13}\le l_{23}$; if $D>0$, then $l_{34}\le 1$, $l_{13}\le l_{23}$}\\
  & \multicolumn{2}{l}{$U_1^{n_1}=U_2^{n_2}=U_3^{n_3}$: $l_{12}=l_{13}\le l_{23}$}\\\cline{2-3}
  & $\sum_{\sigma}\mathpzc{a}_4\Big(U_i^{n_i}(\mathfrak{p}_{\sigma(i)})\big|_{i=1}^4;l_{ij}\big|_{1\le i<j\le 4},(\tau_1\tau_2\tau_3)\Big)$ & $\mathpzc{a}_4\Big(\langle U_i^{n_i}\rangle\big|_{i=1}^4;l_{ij}\big|_{1\le i<j\le 4},(\tau_1\tau_2\tau_3)\Big)$\\\cline{2-3}
  & \multicolumn{2}{l}{$K=n_1+n_2+n_3+n_4$ odd, $K-1\ge 2n_i$, $D=n_1+n_2-n_3-n_4$}\\
  & \multicolumn{2}{l}{$U_i^{n_i}$ mutually unequal: If $D\ge 1$, then $(\tau_1\tau_2\tau_3)=(123),(124)$ with $l_{34}=0$;}\\
  & \multicolumn{2}{l}{\phantom{$U_i^{n_i}$ mutually unequal: }if $D\le -1$, then $(\tau_1\tau_2\tau_3)=(134),(234)$ with $l_{12}=0$}\\
  & \multicolumn{2}{l}{$U_1^{n_1}=U_2^{n_2}$: If $D\ge 1$, then $(\tau_1\tau_2\tau_3)=(123),(124)$ with $l_{34}=0,\ l_{13}< l_{23}$;}\\
  & \multicolumn{2}{l}{\phantom{$U_1^{n_1}=U_2^{n_2}$: }if $D\le -1$, then $(\tau_1\tau_2\tau_3)=(134)$ with $l_{12}=0$}\\
  & \multicolumn{2}{l}{$U_1^{n_1}=U_2^{n_2}=U_3^{n_3}$: $(\tau_1\tau_2\tau_3)=(124)$. If $D\le -1$, $l_{12}=l_{13}<l_{23}$; if $D\ge 1$, $l_{34}=l_{24}<l_{14}$}\\\hline
    \multicolumn{3}{l}{For the notations: $\overset{p}{\cdot}$ and $\overset{p}{\times}$, see \eqref{oversetdef}; $\mathpzc{a}_3$ and $\mathpzc{a}_4$, see \eqref{terms_a3}, \eqref{a4_def}, and \eqref{lrelation}.}
  \end{tabular}
\end{table}

The correspondence of terms in the free energy and the terms in the expansion is crucial for computing their coefficients from the microscopic interaction. 
The molecular potential determines $\mathscr{G}$, then determines $\mathscr{M}$, on which the expansion is done. 
When one attempts to compute the coefficients, the orthogonality can bring convenieces when doing the computation. 
The coefficients are calculated for rod-like \cite{RodModel} and bent-core molecules \cite{BentModel}. 
In these works, the results presented in this section are not utilized, so that lengthy calculation has to be done.

\section{Molecular symmetry\label{expansion_sym}}
Molecular symmetry is characterized by orthogonal transformations that leave the molecule invariant.
Under these transformations, the kernel function $\mathscr{G}_k$ shall also be invariant.
Therefore, the molecular symmetry enforces symmetries on the interaction kernels, thus affects the expansion of these kernels. 
In the previous section, we express the expansion by symmetric traceless tensors. 
This will bring conveniences when discussing molecular symmetry, since the conditions from molecular symmetry are imposed on symmetric traceless tensors. 

All the orthogonal transformations leaving the molecule invariant form a point group $\mathcal{G}$ in $O(3)$, of which all the proper rotations (determinant-one transfomations) form a $SO(3)$-subgroup $\mathcal{G}_1$. 
If $\mathcal{G}$ does not have improper rotations, then $\mathcal{G}_1=\mathcal{G}$. Otherwise, $\mathcal{G}$ can be divided into the union of two cosets, 
\begin{align}
  \mathcal{G}=\mathcal{G}_1\cup (-\mathfrak{k})\mathcal{G}_1=\mathcal{G}_1\cup \mathcal{G}_1(-\mathfrak{k}), 
\end{align}
where $-\mathfrak{k}$ is any improper rotation in $\mathcal{G}$. 
Here, we write the improper rotation as $-\mathfrak{k}$ so that $\mathfrak{k}\in SO(3)$. 

Let us first examine proper rotations. 
For a proper rotation $\mathfrak{s}\in SO(3)$ in the symmetry group $\mathcal{G}$, the kernel function shall be invariant if we rotate any molecule by $\mathfrak{s}$ in the body-fixed frame, i.e. $\mathfrak{p}\to \mathfrak{ps}$. 
Thus, we have
\begin{equation}
  \mathscr{G}_n(\bm{r}_2,\ldots,\bm{r}_n,\mathfrak{p}_1,\ldots,\mathfrak{p}_j\mathfrak{s},\ldots)=\mathscr{G}_n(\bm{r}_2\ldots,\bm{r}_n,\mathfrak{p}_1,\ldots,\mathfrak{p}_j,\ldots). \label{sym}
\end{equation}
It tells us
\begin{equation}
  \mathscr{M}_n^{k_2,\ldots,k_n}(\ldots,\mathfrak{p}_j\mathfrak{s},\ldots)=\mathscr{M}_n^{k_2,\ldots,k_n}(\ldots,\mathfrak{p}_j,\ldots). \label{symM}
\end{equation}

Recall that for the $SO(3)$-subgroup $\mathcal{G}_1$, the $l$-th order symmetric traceless tensors can be decomposed into two orthogonal subspaces $\mathbb{A}^{\mathcal{G}_1,l}$ and $(\mathbb{A}^{\mathcal{G}_1,l})^{\perp}$ \cite{xu_tensors}, such that
\begin{align*}
  &A(\mathfrak{ps})=A(\mathfrak{p}),\qquad\forall A\in\mathbb{A}^{\mathcal{G}_1,l},\ \mathfrak{s}\in\mathcal{G}_1;\\
  &\frac{1}{\#\mathcal{G}_1}\sum_{\mathfrak{s}\in\mathcal{G}_1}A(\mathfrak{ps})=0, \qquad \forall A\in(\mathbb{A}^{\mathcal{G}_1,l})^{\perp}. 
\end{align*}
Since any term $\varPhi\big(U_1(\mathfrak{p}_1),\ldots,U_n(\mathfrak{p}_n)\big)$ in the expansion is multi-linear about $(U_1,\ldots,U_n)$, we have 
\begin{align*}
  &\Big(\mathscr{M}_n^{k_2,\ldots,k_n}(\mathfrak{p}_1,\ldots,\mathfrak{p}_n),\varPhi\big(U_1(\mathfrak{p}_1),\ldots,U_n(\mathfrak{p}_n)\big)\Big)\\
  =&\int \mathscr{M}_n^{k_2,\ldots,k_n}(\mathfrak{p}_1,\ldots,\mathfrak{p}_n)\cdot\varPhi\big(U_1(\mathfrak{p}_1),\ldots,U_n(\mathfrak{p}_n)\big)\,\md\mathfrak{p}_1\ldots\md\mathfrak{p}_n\\
  =&\frac{1}{\# \mathcal{G}_1}\sum_{\mathfrak{s}\in\mathcal{G}_1}
  \int \mathscr{M}_n^{k_2,\ldots,k_n}(\mathfrak{p}_1\mathfrak{s},\ldots,\mathfrak{p}_n)\cdot\varPhi\big(U_1(\mathfrak{p}_1),\ldots,U_n(\mathfrak{p}_n)\big)\,\md\mathfrak{p}_1\ldots\md\mathfrak{p}_n\\
  =&\frac{1}{\# \mathcal{G}_1}\sum_{\mathfrak{s}\in\mathcal{G}_1}
  \int \mathscr{M}_n^{k_2,\ldots,k_n}(\mathfrak{p}_1,\ldots,\mathfrak{p}_n)\cdot\varPhi\big(U_1(\mathfrak{p}_1\mathfrak{s}^{-1}),\ldots,U_n(\mathfrak{p}_n)\big)\,\md(\mathfrak{p}_1\mathfrak{s})\md\mathfrak{p}_2\ldots\md\mathfrak{p}_n\\
  =&\int \mathscr{M}_n^{k_2,\ldots,k_n}(\mathfrak{p}_1,\ldots,\mathfrak{p}_n)\cdot\varPhi\big(\frac{1}{\# \mathcal{G}_1}\sum_{\mathfrak{s}\in\mathcal{G}_1}U_1(\mathfrak{p}_1\mathfrak{s}^{-1}),\ldots,U_n(\mathfrak{p}_n)\big)\,\md\mathfrak{p}_1\ldots\md\mathfrak{p}_n.
\end{align*}
%
Together with the orthogonality of terms (Theorem \ref{expansion_tensor}), the above derivation implies the following theorem. 
\begin{theorem}\label{tensor_proper}
For each term $\varPhi\big(U_1(\mathfrak{p}_1),\ldots,U_n(\mathfrak{p}_n)\big)$ in the expansion, the tensors $U_i$ can only take invariant tensors of $\mathcal{G}_1$. 
\end{theorem}


Next, we discuss improper rotations.
Let us consider the following operations.
For a cluster with $n$ molecules, we inverse them as a whole.
The body-fixed frames are transformed from $(\bm{x}_i,\mathfrak{p}_i)$ into $(-\bm{x}_i,-\mathfrak{p}_i)$. 
The frames are now left-handed, which can be recovered to right-handed ones by an improper rotation $-\mathfrak{k}$.
The final result is 
\begin{align*}
  (\bm{x}_i,\mathfrak{p}_i)\longrightarrow (-\bm{x}_i,\mathfrak{p}_i\mathfrak{k}). 
\end{align*}
Therefore, we obtain 
\begin{equation}
  \mathscr{G}_n(-\bm{r}_2,\ldots,-\bm{r}_n,\mathfrak{p}_1\mathfrak{k},\ldots,\mathfrak{p}_j\mathfrak{k})=\mathscr{G}_n(\bm{r}_2\ldots,\bm{r}_n,\mathfrak{p}_1,\ldots,\mathfrak{p}_n). \label{symimp}
\end{equation}
It tells us
\begin{equation}
  \mathscr{M}_n^{k_2,\ldots,k_n}(\mathfrak{p}_1\mathfrak{k},\ldots,\mathfrak{p}_n\mathfrak{k})=(-1)^{k_2+\ldots +k_n}\mathscr{M}_n^{k_2,\ldots,k_n}(\mathfrak{p}_1\mathfrak{k},\ldots,\mathfrak{p}_n\mathfrak{k}). \label{symMimp}
\end{equation}
Following the same derivation above Theorem \ref{tensor_proper}, we need to examine what the tensors $V(\mathfrak{pk})$ are for the invariant tensors $V(\mathfrak{p})\in \mathbb{A}^{\mathcal{G}_1,l}$. 

\begin{proposition}\label{decomp_improper}
According to the improper rotation $-\mathfrak{k}\in\mathcal{G}$, the space of invariant tensors $\mathbb{A}^{\mathcal{G}_1,l}$ can be decomposed into the sum of two orthogonal subspaces, 
$$
\mathbb{A}^{\mathcal{G},l}_{+1}=\{V(\mathfrak{p})\in\mathbb{A}^{\mathcal{G}_1,l}:V(\mathfrak{pk})=V(\mathfrak{p})\},\quad \mathbb{A}^{\mathcal{G},l}_{-1}=\{V(\mathfrak{p})\in\mathbb{A}^{\mathcal{G}_1,l}:V(\mathfrak{pk})=-V(\mathfrak{p})\}. 
$$
\end{proposition}
\begin{proof}
We shall notice that for any proper rotation $\mathfrak{s}$ in the point group $\mathcal{G}$, we have $\mathfrak{ksk}$ is also a proper rotation in $\mathcal{G}$.
This can be recognized by writing it as $\mathfrak{(-k)s(-k)}$, a composition of three elements in the group, two of which are improper rotations. 

For an invariant tensor $V(\mathfrak{p})$, we can express it as 
\begin{align*}
  V(\mathfrak{p})=\frac{1}{2}(V(\mathfrak{p})+V(\mathfrak{pk}))+\frac{1}{2}(V(\mathfrak{p})-V(\mathfrak{pk})), 
\end{align*}
where $V(\mathfrak{p})+V(\mathfrak{pk})$ is invariant under $\mathfrak{k}$, and $V(\mathfrak{p})-V(\mathfrak{pk})$ is transformed into its opposite $V(\mathfrak{pk})-V(\mathfrak{pk}^2)=V(\mathfrak{pk})-V(\mathfrak{p})$. 
\end{proof}
For the tensors in $\mathbb{A}_{\pm 1}^{\mathcal{G},l}$, we call them tensors of type $\pm 1$. 

Similar to the derivation for Theorem \ref{tensor_proper}, the effect of improper rotations is stated below. 
\begin{theorem}\label{coupling_improper}
  In the expansion of $\mathscr{M}_n^{k_2,\ldots,k_n}$, let $k=k_2+\ldots +k_n$. 
  When $k$ is odd, the tensors of type $-1$ shall appear odd times. 
  When $k$ is even, the tensors of type $-1$ shall appear even times.
\end{theorem}
In particular, for $\mathscr{M}_2^{k}$, when $k$ is even, the coupling shall be between two tensors of type $+1$ or type $-1$; when $k$ is odd, the coupling shall be between one tensor of type $+1$ and one of type $-1$. 
For $\mathscr{M}_3^{0,0}$ and $\mathscr{M}_4^{0,0,0}$, the number of tensors of type $-1$ shall be zero, two or four. 

We pay attention to the case where the group $\mathcal{G}$ has the inversion, i.e. $\mathfrak{k}=\mathfrak{i}$, so that $\mathcal{G}_1\mathfrak{k}=\mathcal{G}_1$. In this case, we have $\mathbb{A}_{+1}^{\mathcal{G},l}=\mathbb{A}^{\mathcal{G}_1,l}$ and $\mathbb{A}_{-1}^{\mathcal{G},l}=\{0\}$.
If the group $\mathcal{G}$ does not include the inversion, we need to identify the two spaces. 

\subsection{Tensors of two types for each point group}
Based on our discussion above, we find out the tensors of type $\pm 1$ for each point group.
The point groups have been identified completely (see, for example, \cite{Group_Cotton}), and the invariant tensors for point groups in $SO(3)$ have been identified in \cite{xu_tensors}.
Thus, our task is to write down the decomposition in Proposition \ref{decomp_improper}. 
For the point groups having the common $SO(3)$-subgroup, we will discuss together and see how they are distinguished by the improper rotations. 

First, let us write down the rotation subgroup and one improper rotation in each point group. 
Recall that the frame fixed on a molecule is $\mathfrak{p}=(\bm{m}_1,\bm{m}_2,\bm{m}_3)$, and a rotation within this frame is expressed by $\mathfrak{p}\to\mathfrak{ps}$.
Let us introduce some rotations below, 
\begin{align}
  &\mathfrak{j}_{\theta}=\left(
  \begin{array}{ccc}
    1 & 0 & 0\\
    0 & \cos\theta & -\sin\theta\\
    0 & \sin\theta & \cos\theta
  \end{array}
  \right), \quad
  \mathfrak{b}_2=\left(
  \begin{array}{ccc}
    -1 & 0 & 0\\
    0 & 1 & 0\\
    0 & 0 & -1
  \end{array}
  \right),
  \mathfrak{r}_3=\left(
  \begin{array}{ccc}
    0 & 0 & 1\\
    1 & 0 & 0\\
    0 & 1 & 0
  \end{array}
  \right),\nonumber\\
  &\mathfrak{v}_5=\frac{1}{2}\left(\begin{array}{ccc}
    \phi & -1 & \phi-1\\
    1 & \phi-1 & -\phi\\
    \phi-1 & \phi & 1
  \end{array}\right),\quad \phi=\frac{1+\sqrt{5}}{2}. 
  \label{basicrot}
\end{align}
In the above, $\mathfrak{j}_{\theta}$ is the rotation round $\bm{m}_1$ by the angle $\theta$. To comprehend this rotation, we could write out 
\begin{align*}
  \mathfrak{ps}=&(\bm{m}_1,\bm{m}_2,\bm{m}_3)\left(
  \begin{array}{ccc}
    1 & 0 & 0\\
    0 & \cos\theta & -\sin\theta\\
    0 & \sin\theta & \cos\theta
  \end{array}
  \right)\\
  =&(\bm{m}_1,\cos\theta\bm{m}_2+\sin\theta\bm{m}_3,-\sin\theta\bm{m}_2+\cos\theta\bm{m}_3). 
\end{align*}
Moreover, for two angles $\theta_1$ and $\theta_2$, we have
$$
\mathfrak{j}_{\theta_1}\mathfrak{j}_{\theta_2}=\mathfrak{j}_{\theta+\theta_2}. 
$$
Thus, for an integer $m$ we have 
\begin{align}
  \mathfrak{j}_{\theta}^m=\mathfrak{j}_{m\theta}. \nonumber
\end{align}
The second one, $\mathfrak{b}_2$, is the rotation round $\bm{m}_2$ by the angle $\pi$; 
$\mathfrak{r}_3$ is the rotation round $(\bm{m}_1+\bm{m}_2+\bm{m}_3)/\sqrt{3}$ by $2\pi/3$, transforming $(\bm{m}_1,\bm{m}_2,\bm{m}_3)$ into $(\bm{m}_2,\bm{m}_3,\bm{m}_1)$; 
and $\mathfrak{v}_5$ is a five-fold rotation. 


For each point group, we explain how to pose the body-fixed frame $(\bm{m}_i)$ and write down the generating elements. 
The generating elements and illustrations can be found in other works, such as \cite{Group_Cotton,xu_softmatter}. 
We shall present in the following way: describe a point group in $SO(3)$ (with only proper rotations); then, for all the groups containing it as the rotation subgroup, we specify an improper rotation $-\mathfrak{k}$. 
\begin{itemize}
\item The group $\mathcal{C}_{\infty}$ consists of rotations round an axis by arbitrary angle. We choose $\bm{m}_1$ as the axis, so that $\mathcal{C}_{\infty}=\{\mathfrak{j}_{\theta},\forall \theta\}$. 
\begin{itemize}
\item $\mathcal{C}_{\infty v}$ has a mirror plane $\hat{O}\bm{m}_1\bm{m}_2$, so an improper rotation is $-\mathfrak{k}=\mathrm{diag}(1,1,-1)=-\mathfrak{j}_{\pi}\mathfrak{b}_2$. 
\item $\mathcal{C}_{\infty h}$ has a mirror plane $\hat{O}\bm{m}_2\bm{m}_3$, so an improper rotation is $\mathrm{diag}(-1,1,1)=-\mathfrak{j}_{\pi}$. We multiply it with a proper rotation $\mathfrak{j}_{\pi}$ to recognize that the inversion $-\mathfrak{i}$ belongs to $\mathcal{C}_{\infty h}$. 
\end{itemize}
\item The group $\mathcal{D}_{\infty}$ contains $\mathcal{C}_{\infty}$ as a subset, and also allows $\mathfrak{b}_2$.
\begin{itemize}
\item $\mathcal{D}_{\infty h}$ has a mirror plane $\hat{O}\bm{m}_2\bm{m}_3$, so it contains the inversion. 
\end{itemize}
\item $\mathcal{C}_n$ is generated by the rotation round $\bm{m}_1$ by the angle $2\pi/n$, i.e. is generated by $\mathfrak{j}_{2\pi/n}$. 
\begin{itemize}
\item $\mathcal{C}_{nv}$ has an improper rotation $\mathrm{diag}(1,1,-1)=-\mathfrak{j}_{\pi}\mathfrak{b}_2$. 
\item $\mathcal{C}_{nh}$ has an improper rotation $\mathrm{diag}(-1,1,1)=-\mathfrak{j}_{\pi}$. When $n$ is even, we multiply it by $\mathfrak{j}_{2\pi/n}^{n/2}=\mathfrak{j}_{\pi}$ to get the inversion. When $n$ is odd, we multiply it by $\mathfrak{j}_{2\pi/n}^{(n+1)/2}=\mathfrak{j}_{(n+1)\pi/n}$ and let $\mathfrak{k}=\mathfrak{j}_{\pi/n}$. 
\item $\mathcal{S}_{2n}$ allows a roto-reflection round $\bm{m}_1$, i.e. to rotate round $\bm{m}_1$ by the angle $\pi/n$, followed by a reflection about the plane $\hat{O}\bm{m}_2\bm{m}_3$. Such an improper rotation can be expressed by $\mathfrak{j}_{\pi/n}(-\mathfrak{j}_{\pi})=-\mathfrak{j}_{(n+1)\pi/n}$. When $n$ is odd, we multiply it by $\mathfrak{j}_{2\pi/n}^{(n-1)/2}=\mathfrak{j}_{(n-1)\pi/n}$ to get the inversion. When $n$ is even, we multiply it by $\mathfrak{j}_{2\pi/n}^{n/2}=\mathfrak{j}_{\pi}$ and let $\mathfrak{k}=\mathfrak{j}_{\pi/n}$. 
\end{itemize}
\item $\mathcal{D}_n$ is generated by $\mathfrak{j}_{2\pi/n}$ and $\mathfrak{b}_2$. 
\begin{itemize}
\item $\mathcal{D}_{nh}$ has an improper rotation $\mathrm{diag}(-1,1,1)=-\mathfrak{j}_{\pi}$. When $n$ is even, the group contains the inversion. When $n$ is odd, we let $\mathfrak{k}=\mathfrak{j}_{\pi/n}$. 
\item $\mathcal{D}_{nd}$ has an improper rotation $-\mathfrak{j}_{\pi}\mathfrak{j}_{\pi/n}=-\mathfrak{j}_{(n+1)\pi/n}$. When $n$ is odd, the group contains the inversion. When $n$ is even, we let $\mathfrak{k}=\mathfrak{j}_{\pi/n}$. 
\end{itemize}
\item $\mathcal{T}$ contains all the proper rotations allowed by a regular tetrahedron, which can be generated by $\mathfrak{j}_{\pi}$, $\mathfrak{b}_2$ and $\mathfrak{r}_3$. 
\begin{itemize}
\item $\mathcal{T}_d$ allows the improper rotation 
  \begin{align*}
    \left(
    \begin{array}{ccc}
      1 & 0 & 0\\
      0 & 0 & 1\\
      0 & 1 & 0
    \end{array}
    \right)=\mathfrak{j}_{\pi/2}\mathrm{diag}(1,1,-1)=-\mathfrak{j}_{3\pi/2}\mathfrak{b}_2.
  \end{align*}
  We multiply it by the proper rotation $\mathfrak{j}_{\pi}\mathfrak{b}_2$ in $\mathcal{T}$, so that we may let $\mathfrak{k}=\mathfrak{j}_{\pi/2}$. 
\item $\mathcal{T}_h$ has a mirror plane $\hat{O}\bm{m}_2\bm{m}_3$, so it contains the inversion. 
\end{itemize}
\item $\mathcal{O}$ contains all the proper rotations allowed by a cube, which can be generated by $\mathfrak{j}_{\pi/2}$, $\mathfrak{b}_2$ and $\mathfrak{r}_3$.
\begin{itemize}
\item $\mathcal{O}_h$ contains all the $O(3)$ transformations of a cube, allowing the inversion. 
\end{itemize}
\item $\mathcal{I}$ contains all the proper rotations allowed by a regular icosahedron, generated by $\mathfrak{j}_{\pi}$, $\mathfrak{b}_2$, $\mathfrak{r}_3$, $\mathfrak{v}_5$.
\begin{itemize}
\item $\mathcal{I}_h$ contains all the $O(3)$ transformations of a regular icosahedron, allowing the inversion. 
\end{itemize}
\end{itemize}

For each point group in $SO(3)$, we write down the invariant tensors obtained in \cite{xu_tensors}, then find out the two types of tensors using the improper rotations. 
To express symmetric traceless tensors, we introduce the polynomials 
\begin{align}
  \tilde{T}_n(y,z)=z^{n/2}T_n(y/\sqrt{z}),\ \tilde{U}_n(y,z)=z^{n/2}U_n(y/\sqrt{z}),\ \tilde{P}_n^{(\mu,\mu)}(y,z)=z^{n/2}P_n^{(\mu,\mu)}(y/\sqrt{z}), 
\end{align}
where $T_n(\cos\theta)=\cos n\theta$ and $U_{n-1}(\cos\theta)\sin\theta=\sin n\theta$ are the Chebyshev polynomials of the first and the second kind, and $P_n^{(\mu,\mu)}(x)$ is the Jacobi polynomial with two identical indices $(\mu,\mu)$.
Since the Chebyshev and Jacobi polynomials only have the terms with the same parity as the order $n$ (see Appendix for explicit expressions), the above definition indeed gives polynomials of $y$ and $z$. 
According to the monomial notation \eqref{tensor_monomial}, when we substitute $y,z$ by some polynomials of $\bm{m}_i$, we define a symmetric tensor.

For all the point groups having improper rotations, the tensors of type $\pm 1$ are listed in Table \ref{tensor_improper}, which we explain below. 

\begin{table}
  \centering
  \caption{Tensors of type $\pm 1$ for point groups containing improper rotations. \label{tensor_improper}}\small
  \begin{tabular}{c|l}\hline
  Group & Tensors of types $\pm 1$, two spaces $\mathbb{A}^{\mathcal{G},l}_{+1}$ and $\mathbb{A}^{\mathcal{G},l}_{-1}$\\\hline
  $\mathcal{C}_{\infty h},\mathcal{D}_{\infty h}$ & \\
  $\mathcal{C}_{nh}$ ($n$ even) & Improper rotation $-\mathfrak{k}=-\mathfrak{i}$\\
  $\mathcal{S}_{2n}$ ($n$ odd) & $\mathcal{G}=\mathcal{G}_1\cup(-\mathcal{G}_1)$, $\mathcal{G}_1$ rotation subgroup \\
  $\mathcal{D}_{nh}$ ($n$ even) & $\mathbb{A}^{\mathcal{G},l}_{+1}=\mathbb{A}^{\mathcal{G}_1,l}$, $\mathbb{A}^{\mathcal{G},l}_{-1}=\{0\}$\\
  $\mathcal{D}_{nd}$ ($n$ odd) &see \eqref{tensors_Cinf}, \eqref{tensors_Dinf}, \eqref{tensors_Cn}, \eqref{tensors_Dn}, \eqref{tensors_poly1}, \eqref{tensors_poly2}, \eqref{tensors_poly3}\\
  $\mathcal{T}_h,\mathcal{O}_h,\mathcal{I}_h$ &\\\hline
  $\mathcal{C}_{\infty v}$ & $l$ even, 
  $+1:\mathrm{span}\left\{\tilde{P}_{l}^{(0,0)}(\bm{m}_1,\mathfrak{i})\right\}$;\\
  & \phantom{$l$ even,}  $-1:\{0\}$\\ 
  & $l$ odd,   $+1:\{0\}$; \\
  & \phantom{$l$ odd,}  $-1:\mathrm{span}\left\{\tilde{P}_{l}^{(0,0)}(\bm{m}_1,\mathfrak{i})\right\}$\\\hline
  $\mathcal{C}_{nv}$ & $l$ even, 
  $+1:\mathrm{span}\left\{\tilde{P}_{l-jn}^{(jn,jn)}(\bm{m}_1,\mathfrak{i})\tilde{T}_{jn}(\bm{m}_2,\mathfrak{i}-\bm{m}_1^2)\right\}$\\
  & \phantom{$l$ even,}
  $-1:\mathrm{span}\left\{\tilde{P}_{l-jn}^{(jn,jn)}(\bm{m}_1,\mathfrak{i})\tilde{U}_{jn-1}(\bm{m}_2,\mathfrak{i}-\bm{m}_1^2)\bm{m}_3\right\}$\\ 
  & $l$ odd, 
  $+1:\mathrm{span}\left\{\tilde{P}_{l-jn}^{(jn,jn)}(\bm{m}_1,\mathfrak{i})\tilde{U}_{jn-1}(\bm{m}_2,\mathfrak{i}-\bm{m}_1^2)\bm{m}_3\right\}$\\
  & \phantom{$l$ odd,}
  $-1:\mathrm{span}\left\{\tilde{P}_{l-jn}^{(jn,jn)}(\bm{m}_1,\mathfrak{i})\tilde{T}_{jn}(\bm{m}_2,\mathfrak{i}-\bm{m}_1^2)\right\}$ \\\hline
  $\mathcal{S}_{2n}$ ($n$ even) 
  & $+1:\mathrm{span}\left\{\tilde{P}_{l-jn}^{(jn,jn)}(\bm{m}_1,\mathfrak{i})\tilde{T}_{jn}(\bm{m}_2,\mathfrak{i}-\bm{m}_1^2),\tilde{P}_{l-jn}^{(jn,jn)}(\bm{m}_1,\mathfrak{i})\tilde{U}_{jn-1}(\bm{m}_2,\mathfrak{i}-\bm{m}_1^2)\bm{m}_3,\ j \text{ even} \right\}$\\
  $\mathcal{C}_{nh}$ ($n$ odd)
  & $-1:\mathrm{span}\left\{\tilde{P}_{l-jn}^{(jn,jn)}(\bm{m}_1,\mathfrak{i})\tilde{T}_{jn}(\bm{m}_2,\mathfrak{i}-\bm{m}_1^2),\tilde{P}_{l-jn}^{(jn,jn)}(\bm{m}_1,\mathfrak{i})\tilde{U}_{jn-1}(\bm{m}_2,\mathfrak{i}-\bm{m}_1^2)\bm{m}_3,\ j \text{ odd} \right\}$\\\hline
  $\mathcal{D}_{nd}$ ($n$ even)
  & $l$ even, 
  $+1:\mathrm{span}\left\{\tilde{P}_{l-jn}^{(jn,jn)}(\bm{m}_1,\mathfrak{i})\tilde{T}_{jn}(\bm{m}_2,\mathfrak{i}-\bm{m}_1^2),\ j \text{ even}\right\}$\\
  & \phantom{$l$ even,}
  $-1:\mathrm{span}\left\{\tilde{P}_{l-jn}^{(jn,jn)}(\bm{m}_1,\mathfrak{i})\tilde{T}_{jn}(\bm{m}_2,\mathfrak{i}-\bm{m}_1^2),\ j \text{ odd}\right\}$\\
  & $l$ odd,
  $+1:\mathrm{span}\left\{\tilde{P}_{l-jn}^{(jn,jn)}(\bm{m}_1,\mathfrak{i})\tilde{U}_{jn-1}(\bm{m}_2,\mathfrak{i}-\bm{m}_1^2)\bm{m}_3,\ j \text{ even} \right\}$\\
  & \phantom{$l$ odd,}
  $-1:\mathrm{span}\left\{\tilde{P}_{l-jn}^{(jn,jn)}(\bm{m}_1,\mathfrak{i})\tilde{U}_{jn-1}(\bm{m}_2,\mathfrak{i}-\bm{m}_1^2)\bm{m}_3,\ j \text{ odd} \right\}$\\\hline
  $\mathcal{D}_{nh}$ ($n$ odd)
  & $l$ even,
  $+1:\mathrm{span}\left\{\tilde{P}_{l-jn}^{(jn,jn)}(\bm{m}_1,\mathfrak{i})\tilde{T}_{jn}(\bm{m}_2,\mathfrak{i}-\bm{m}_1^2),\ j \text{ even}\right\}$\\
  & \phantom{$l$ even,}
  $-1:\mathrm{span}\left\{\tilde{P}_{l-jn}^{(jn,jn)}(\bm{m}_1,\mathfrak{i})\tilde{U}_{jn-1}(\bm{m}_2,\mathfrak{i}-\bm{m}_1^2)\bm{m}_3,\ j \text{ odd} \right\}$\\
  & $l$ odd, 
  $+1:\mathrm{span}\left\{\tilde{P}_{l-jn}^{(jn,jn)}(\bm{m}_1,\mathfrak{i})\tilde{U}_{jn-1}(\bm{m}_2,\mathfrak{i}-\bm{m}_1^2)\bm{m}_3,\ j \text{ even} \right\}$\\
  & \phantom{$l$ odd,}
  $-1:\mathrm{span}\left\{\tilde{P}_{l-jn}^{(jn,jn)}(\bm{m}_1,\mathfrak{i})\tilde{T}_{jn}(\bm{m}_2,\mathfrak{i}-\bm{m}_1^2),\ j \text{ odd}\right\}$ \\\hline
  $\mathcal{T}_d$ & $+1:\mathrm{span}\bigg\{\Big\{\big(S_2^i(\bm{m}_1,\bm{m}_2,\bm{m}_3)S_3^j(\bm{m}_1,\bm{m}_2,\bm{m}_3)\big)_0,\ j\text{ even}$, $l=4i+3j\Big\}$\\
  (see \eqref{sympols}) & $\ \ \ \cup\Big\{\big(E(\bm{m}_1,\bm{m}_2,\bm{m}_3)S_2^i(\bm{m}_1,\bm{m}_2,\bm{m}_3)S_3^j(\bm{m}_1,\bm{m}_2,\bm{m}_3)\big)_0,\ j\text{ odd}$, $l=6+4i+3j\Big\}\bigg\}$\\
  & $-1:\mathrm{span}\bigg\{\Big\{\big(S_2^i(\bm{m}_1,\bm{m}_2,\bm{m}_3)S_3^j(\bm{m}_1,\bm{m}_2,\bm{m}_3)\big)_0,\ j\text{ odd}$, $l=4i+3j\Big\}$\\
  & $\ \ \ \cup\Big\{\big(E(\bm{m}_1,\bm{m}_2,\bm{m}_3)S_2^i(\bm{m}_1,\bm{m}_2,\bm{m}_3)S_3^j(\bm{m}_1,\bm{m}_2,\bm{m}_3)\big)_0,\ j\text{ even}$, $l=6+4i+3j\Big\}\bigg\}$\\\hline
  \end{tabular}
\end{table}

\subsubsection{Axisymmetries}
We first look into two rotation groups $\mathcal{C}_{\infty}$, $\mathcal{D}_{\infty}$. The invariant tensors are given by 
\begin{align}
  \mathbb{A}^{\mathcal{C}_{\infty},l}=&\mathrm{span}\left\{\tilde{P}_{l}^{(0,0)}(\bm{m}_1,\mathfrak{i})\right\}, \label{tensors_Cinf}\\
  \mathbb{A}^{\mathcal{D}_{\infty},l}=&\mathrm{span}\left\{\tilde{P}_{l}^{(0,0)}(\bm{m}_1,\mathfrak{i})\right\},\ l\text{ even};\qquad \mathbb{A}^{\mathcal{D}_{\infty},l}=\{0\},\ l\text{ odd}. \label{tensors_Dinf}
\end{align}

For the groups $\mathcal{C}_{\infty h}$, $\mathcal{D}_{\infty h}$, since they possess the inversion, the type $+1$ tensors are just the invariant tensors, and the only type $-1$ tensor is the zero tensor. 

For $\mathcal{C}_{\infty v}$, we have chosen $\mathfrak{k}=\mathrm{diag}(-1,-1,1)$.
Thus, in type $+1$ tensors, $\bm{m}_1$ shall appear even times, while in type $-1$ tensors, $\bm{m}_1$ shall appear odd times. 
As a result, the type $+1$ tensors are those whose order $l$ are even, and the type $-1$ tensors are those with odd order. 

\subsubsection{Finite order axial symmetries}
Next, we look into point groups with the rotation subgroup $\mathcal{C}_n$ or $\mathcal{D}_n$.

The group $\mathcal{C}_n$ is the rotation subgroup of $\mathcal{C}_{nv}$, $\mathcal{C}_{nh}$ and $\mathcal{S}_{2n}$.
The invariant tensors for $\mathcal{C}_n$ are 
\begin{align}
  \mathbb{A}^{\mathcal{C}_n,l}=&\mathrm{span}\left\{\tilde{P}_{l-jn}^{(jn,jn)}(\bm{m}_1,\mathfrak{i})\tilde{T}_{jn}(\bm{m}_2,\mathfrak{i}-\bm{m}_1^2),\tilde{P}_{l-jn}^{(jn,jn)}(\bm{m}_1,\mathfrak{i})\tilde{U}_{jn-1}(\bm{m}_2,\mathfrak{i}-\bm{m}_1^2)\bm{m}_3\right\}. \label{tensors_Cn}
\end{align}
\begin{itemize}
\item
For $\mathcal{C}_{nv}$, we have chosen $\mathfrak{k}=\mathrm{diag}(-1,-1,1)$. 
Therefore, in type $+1$ tensors $\bm{m}_1$ and $\bm{m}_2$ shall appear even times in total, while if $\bm{m}_1$ and $\bm{m}_2$ appear odd times in total, the tensors are type $-1$. 
According to this requirement, the two types of tensors are given in Table \ref{tensor_improper}. 
\item
For $\mathcal{S}_{2n}$ where $n$ is odd, and $\mathcal{C}_{nh}$ where $n$ is even, these groups have the inversion. 
\item
For $\mathcal{S}_{2n}$ where $n$ is even, and $\mathcal{C}_{nh}$ where $n$ is odd, we have chosen $\mathfrak{k}=\mathfrak{j}_{\pi/n}$. 
Now, we use the fact that 
\begin{align}
  (\bm{m}_2+\sqrt{-1}\bm{m}_3)^n=\tilde{T}_n(\bm{m}_2,\mathfrak{i}-\bm{m}_1^2)+\sqrt{-1}\tilde{U}_{n-1}(\bm{m}_2,\mathfrak{i}-\bm{m}_1^2)\bm{m}_3. \label{Chebyrel}
\end{align}
We substitute $\bm{m}_i$ with $\bm{m}_i(\mathfrak{pj}_{\theta})$ in the above.
The left-hand side gives 
\begin{align*}
  \big(\bm{m}_2(\mathfrak{pj}_{\theta})+\sqrt{-1}\bm{m}_3(\mathfrak{pj}_{\theta})\big)^n=e^{\sqrt{-1}n\theta}(\bm{m}_2+\sqrt{-1}\bm{m}_3)^n. 
\end{align*}
Let $\theta=\pi/n$. We obtain 
\begin{align*}
&\tilde{T}_{jn}\big(\bm{m}_2(\mathfrak{pj}_{\pi/n}),\mathfrak{i}-\bm{m}_1^2(\mathfrak{pj}_{\pi/n})\big)=(-1)^j\tilde{T}_{jn}\big(\bm{m}_2,\mathfrak{i}-\bm{m}_1^2\big),\\
&\tilde{U}_{jn-1}\big(\bm{m}_2(\mathfrak{pj}_{\pi/n}),\mathfrak{i}-\bm{m}_1^2(\mathfrak{pj}_{\pi/n})\big)\bm{m}_3(\mathfrak{pj}_{\pi/n})=(-1)^j\tilde{U}_{jn-1}(\bm{m}_2,\mathfrak{i}-\bm{m}_1^2)\bm{m}_3. 
\end{align*}
Therefore, type $+1$ tensors are those in \eqref{tensors_Cn} where $j$ is even, and type $-1$ tensors are those where $j$ is odd. 
\end{itemize}

We turn to the point groups having the rotation subgroup $\mathcal{D}_n$.
The invariant tensors of $\mathcal{D}_n$ are given by 
\begin{align}
  \mathbb{A}^{\mathcal{D}_n,l}=\mathrm{span}\bigg\{&\left\{\tilde{P}_{l-jn}^{(jn,jn)}(\bm{m}_1,\mathfrak{i})\tilde{T}_{jn}(\bm{m}_2,\mathfrak{i}-\bm{m}_1^2),\ l-jn\text{ even}\right\} \nonumber\\
  &\cup\left\{\tilde{P}_{l-jn}^{(jn,jn)}(\bm{m}_1,\mathfrak{i})\tilde{U}_{jn-1}(\bm{m}_2,\mathfrak{i}-\bm{m}_1^2)\bm{m}_3,\ l-jn\text{ odd}\right\}\bigg\}. \label{tensors_Dn}
\end{align}
\begin{itemize}
\item The two groups, $\mathcal{D}_{nd}$ where $n$ is odd, and $\mathcal{D}_{nh}$ where $n$ is even, contain the inversion. 
\item For $\mathcal{D}_{nd}$ where $n$ is even, and $\mathcal{D}_{nh}$ where $n$ is odd, the discussion is similar to $\mathcal{S}_{2n}$ and $\mathcal{C}_{nh}$. By choosing $\mathfrak{k}=\mathfrak{j}_{\pi/n}$, we conclude that type $+1$ tensors are those with even $j$, and type $-1$ tensors are those with odd $j$. 
\end{itemize}

\subsubsection{Polyhedral symmetries}
There are three polyhedral rotation groups, $\mathcal{T}$, $\mathcal{O}$, $\mathcal{I}$. Define 
\begin{subequations}\label{sympols}
\begin{align}
  S_2(\bm{m}_1,\bm{m}_2,\bm{m}_3)=&\bm{m}_1^2\bm{m}_2^2+\bm{m}_2^2\bm{m}_3^2+\bm{m}_3^2\bm{m}_1^2, \\
  S_3(\bm{m}_1,\bm{m}_2,\bm{m}_3)=&\bm{m}_1\bm{m}_2\bm{m}_3, \\
  E(\bm{m}_1,\bm{m}_2,\bm{m}_3)=&(\bm{m}_1^2-\bm{m}_2^2)(\bm{m}_2^2-\bm{m}_3^2)(\bm{m}_3^2-\bm{m}_1^2). 
\end{align}
\end{subequations}
Using these notations, the invariant tensors are given by 
\begin{align}
  \mathbb{A}^{\mathcal{T},l}=&\mathrm{span}\bigg\{\Big\{\big(S_2^i(\bm{m}_1,\bm{m}_2,\bm{m}_3)S_3^j(\bm{m}_1,\bm{m}_2,\bm{m}_3)\big)_0,\ l=4i+3j\Big\}\nonumber\\
  &\cup\Big\{\big(E(\bm{m}_1,\bm{m}_2,\bm{m}_3)S_2^i(\bm{m}_1,\bm{m}_2,\bm{m}_3)S_3^j(\bm{m}_1,\bm{m}_2,\bm{m}_3)\big)_0,\ l=6+4i+3j\Big\}\bigg\}.\label{tensors_poly1}\\
  \mathbb{A}^{\mathcal{O},l}=&\mathrm{span}\bigg\{\Big\{\big(S_2^i(\bm{m}_1,\bm{m}_2,\bm{m}_3)S_3^j(\bm{m}_1,\bm{m}_2,\bm{m}_3)\big)_0,\ j\text{ even}, l=4i+3j\Big\}\nonumber\\
  &\cup\Big\{\big(E(\bm{m}_1,\bm{m}_2,\bm{m}_3)S_2^i(\bm{m}_1,\bm{m}_2,\bm{m}_3)S_3^j(\bm{m}_1,\bm{m}_2,\bm{m}_3)\big)_0,\ j\text{ odd}, l=6+4i+3j\Big\}\bigg\}.\label{tensors_poly2}\\
  \mathbb{A}^{\mathcal{I},l}=&\{V(\mathfrak{p})\in\mathbb{A}^{\mathcal{T},l}: V(\mathfrak{pv}_5)=V(\mathfrak{p})\}. \label{tensors_poly3}
\end{align}
Here, we recall that $(U)_0$ is the symmetric traceless tensor generated by $U$ (see Proposition \ref{U0}).
If explicit expressions are needed, one could expand the tensors into momials and use the explicit expressions of $(\bm{m}_1^{i_1}\bm{m}_2^{i_2}\bm{m}_3^{i_3})_0$ that are provided in \cite{xu_tensors}.

The three point groups $\mathcal{T}_h$, $\mathcal{O}_h$, $\mathcal{I}_h$ contain the inversion, so nothing needs to be discussed. 

For the group $\mathcal{T}_d$, we have chosen $\mathfrak{k}=\mathfrak{j}_{\pi/2}$.
Because $\mathfrak{j}_{\pi/2}^2=\mathfrak{j}_{\pi}$, it is noticed from generating element that $\mathcal{T}\cup\mathcal{T}\mathfrak{j}_{\pi/2}=\mathcal{O}$. 
Therefore, the type $+1$ tensors for $\mathcal{T}_d$ are just the invariant tensors of $\mathcal{O}$.

\section{Summary and examples\label{summary}}
In this paper, we discuss the expansion of $\mathscr{M}_l^{k_2,\ldots,k_l}(\mathfrak{p}_1,\ldots,\mathfrak{p}_l)$ defined from interaction kernels $\mathscr{G}_k$ that are functions of molecular potential. 
The expansion is expressed by symmetric traceless tensors and is consistent with symmetry arguments, including the translations, rotations and label permutations of the whole cluster, and the molecular symmetry described by a point group.
The orthogonality of terms and the basic approximation result are established, which can be useful if the coefficients need to be calculated from microscopic potential. 

The form of expansion is summarized in two tables presented in the main text. 
If one would like to write down the expansion for certain point group, the procedure below can be followed:
\begin{enumerate}[1)]
\item Choose tensors from the invariant tensors of the rotation subgroup. 
\item Use Table \ref{tensor_improper} to identify the types $\pm 1$ of these tensors. 
\item Insert these tensors into the terms in Table \ref{sum_term}. Notice that Theorem \ref{coupling_improper} gives the conditions on the how many times the type $-1$ tensors shall appear. 
\end{enumerate}

We illustrate the procedure by a couple of examples. 
Consider two point groups $\mathcal{C}_{2v}$ and $\mathcal{S}_4$, both having the rotation subgroup $\mathcal{C}_2$.
The invariant tensors up to second order are picked up: $1$ (zeroth order tensor), $\bm{m}_1$, $\bm{m}_1^2-\frac{1}{3}\mathfrak{i}$, $\bm{m}_2^2-\bm{m}_3^2$, $\bm{m}_2\bm{m}_3$. 
Then, from Table \ref{tensor_improper}, we find out the type $\pm 1$ for each tensor: 
\begin{center}
\begin{tabular}{c|c|c|c|c|c}
  \hline
  & $1$ & $\bm{m}_1$ & $\bm{m}_1^2-\frac{1}{3}\mathfrak{i}$ & $\bm{m}_2^2-\bm{m}_3^2$ & $\bm{m}_2\bm{m}_3$\\\hline
$\mathcal{C}_{2v}$ & $+1$ & $-1$ & $+1$ & $+1$ & $-1$ \\\hline
$\mathcal{S}_{4}$ & $+1$ & $+1$ & $+1$ & $-1$ & $-1$\\\hline
\end{tabular}
\end{center}
For the terms in Table \ref{sum_term}, substitute the tensors in these terms by these five tensors, with noticing Theorem \ref{coupling_improper}. 
For example, let us look at the term $U^n(\mathfrak{p}_1)\overset{n-1}{\times} V^n(\mathfrak{p}_2)+V^n(\mathfrak{p}_1)\overset{n-1}{\times} U^n(\mathfrak{p}_2)$ in $\mathscr{M}_2^1$.
The tensor order shall be equal for $U^n$ and $V^n$ with $n\ge 1$.
Since we choose tensors up to second order, we have $n=1$ or $2$. 
If $n=1$, the only first order invariant tensor above is $\bm{m}_1$.
But we cannot let $U^n(\mathfrak{p})=V^n(\mathfrak{p})=\bm{m}_1$, since one of $U^n$ and $V^n$ needs to be type $+1$ while the other is type $-1$. 
When $n=2$, 
for the group $\mathcal{C}_{2v}$, there are two choices $(U^n,V^n)=(\bm{m}_1^2-\frac{1}{3}\mathfrak{i},\bm{m}_2\bm{m}_3)$ or $(\bm{m}_2^2-\bm{m}_3^2,\bm{m}_2\bm{m}_3)$.
For the group $\mathcal{S}_4$, there are two different choices $(U^n,V^n)=(\bm{m}_1^2-\frac{1}{3}\mathfrak{i},\bm{m}_2^2-\bm{m}_3^2)$ or $(\bm{m}_1^2-\frac{1}{3}\mathfrak{i},\bm{m}_2\bm{m}_3)$. 
The difference originates from the improper rotations in $\mathcal{C}_{2v}$ and $\mathcal{S}_4$, which assign different type $\pm 1$ for the tensors. 

\appendix
\section{Chebyshev and Jacobi polynomials}
The Chebyshev polynomials of the first and second kind can be given by 
\begin{align}
  T_n(x)=\sum_{2j\le n}{n\choose 2j}(x^2-1)^kx^{n-2j},\quad 
  U_n(x)=\sum_{2j\le n}{n+1\choose 2j+1}(x^2-1)^kx^{n-2j}. 
\end{align}
The Jacobi polynomials $P_{n}^{\mu,\mu}$, where the two indices are equal, can be given by 
\begin{align}
  P_{n}^{(\mu,\mu)}(x)
  =&\frac{\Gamma(2\mu+1)\Gamma(n+\mu+1)}{\Gamma(\mu+1)\Gamma(n+2\mu+1)}\sum_{2j\le n}(-1)^j\frac{\Gamma(n-j+\mu+1/2)}{\Gamma(\mu+1/2)j!(n-2j)!}2^{n-2j}x^{n-2j}, 
\end{align}
where $\Gamma$ is the gamma function. 
It is clear that these polynomials have either odd order terms only, or even order terms only. 

\section{Differential operators on $SO(3)$}
The matrix $\mathfrak{p}\in SO(3)$ can be parameterized by three Euler angles, $\alpha$, $\beta$ and $\gamma$, 
\begin{align}
\mathfrak{p}=\left(
\begin{array}{ccc}
 \cos\alpha &\quad -\sin\alpha\cos\gamma &\quad\sin\alpha\sin\gamma\\
 \sin\alpha\cos\beta &\quad\cos\alpha\cos\beta\cos\gamma-\sin\beta\sin\gamma &
 \quad -\cos\alpha\cos\beta\sin\gamma-\sin\beta\cos\gamma\\
 \sin\alpha\sin\beta &\quad\cos\alpha\sin\beta\cos\gamma+\cos\beta\sin\gamma &
 \quad -\cos\alpha\sin\beta\sin\gamma+\cos\beta\cos\gamma
\end{array}
\right),\nonumber
\end{align}
where $0\le\alpha\le \pi$, $0\le\beta,\gamma <2\pi$. 
Then, the operators $\mathcal{L}_i$ can be written as 
\begin{align}
\mathcal{L}_1&=\frac{\partial}{\partial\gamma},\nonumber\\
\mathcal{L}_2&=\frac{-\cos\gamma}{\sin\alpha}\left(\frac{\partial}{\partial\beta}
-\cos\alpha\frac{\partial}{\partial\gamma}\right)
+\sin\gamma\frac{\partial}{\partial\alpha},\nonumber\\
\mathcal{L}_3&=\frac{\sin\gamma}{\sin\alpha}\left(\frac{\partial}{\partial\beta}
-\cos\alpha\frac{\partial}{\partial\gamma}\right)
+\cos\gamma\frac{\partial}{\partial\alpha}. \nonumber
\end{align}
We could verify that 
\begin{equation}
  \mathcal{L}_i\bm{m}_j=\epsilon_{ijk}\bm{m}_k. \label{diffm}
\end{equation}
When acting on a tensor $U(\mathfrak{p})$, the operators $\mathcal{L}_i$ keep the symmetric traceless property. 
Moreover, if we write $(U)_0=U-\mathfrak{i}U_1$, we can deduce that 
\begin{align}
  \mathcal{L}_i\big(U(\mathfrak{p})\big)_0=\mathcal{L}_iU(\mathfrak{p})-\mathfrak{i}\mathcal{L}_iU_1(\mathfrak{p})
  =\big(\mathcal{L}_iU(\mathfrak{p})\big)_0, \label{diffU0}
\end{align}
where the last equality uses Proposition \ref{U0} and the fact that $\mathcal{L}_iU(\mathfrak{p})-\mathfrak{i}\mathcal{L}_iU_1(\mathfrak{p})$ is a symmetric traceless tensor.

\section{Group representation}


Let us consider the space of all $n$-th order symmetric traceless tensors.
The element $\mathfrak{p}\in SO(3)$ acting on $n$-th order symmetric traceless tensors actually defines a linear transformation, because $\mathfrak{p}\circ (\lambda_1U_1+\lambda_2U_2)=\lambda_1\mathfrak{p}\circ U_1+\lambda_2\mathfrak{p}\circ U_2$. 
Moreover, $U(\mathfrak{p}_1\mathfrak{p}_2)=\mathfrak{p}_1\circ U(\mathfrak{p}_2)$ implies that the map from $\mathfrak{p}$ to the linear transformation is a group representation. 
The equality \eqref{innrot} that the rotation keeps the dot product indicates that this representation is unitary. 

If a subspace $\mathbb{V}$ satisfies $\mathfrak{p}\circ V\in\mathbb{V}$ for any $V\in\mathbb{V}$, it is called an invariant subspace.
If the only invariant subspaces are zero space and the whole space, then the representation is called irreducible. 
In fact, the representation defined by $\mathfrak{p}\circ U$ on $n$-th order symmetric traceless tensors is irreducible. 
Then, applying the theory of group representation (see, for example, \cite{SpecFun}), Proposition \ref{orth_basis_mat} is established. 

The irreducibility is claimed previously (see, for example, \cite{Group_Boerner}), but the derivation might be presented for other mathematical objects only. 
We give a brief note below. 
Any invariant subspace is also invariant under the differential operators. 
To show irreducibility, we start from any nonzero tensor to generate a basis by taking derivatives. 
Direct calculation using \eqref{diffm} yields 
\begin{align*}
  &(\sqrt{1}\mathcal{L}_2\pm\mathcal{L}_3)\big(\bm{m}_1^{n-k}(\bm{m}_2+\sqrt{-1}\bm{m}_3)^k\big)=(-n\mp k)\bm{m}_1^{n-k\pm 1}(\bm{m}_2+\sqrt{-1}\bm{m}_3)^{k\mp 1}.
\end{align*}
Together with \eqref{diffU0}, we arrive at 
\begin{align*}
  &(\sqrt{1}\mathcal{L}_2\pm\mathcal{L}_3)\big(\bm{m}_1^{n-k}(\bm{m}_2+\sqrt{-1}\bm{m}_3)^k\big)_0=(-n\mp k)\big(\bm{m}_1^{n-k\pm 1}(\bm{m}_2+\sqrt{-1}\bm{m}_3)^{k\mp 1}\big)_0.
\end{align*}
Note that $\big(\bm{m}_1^{n-k}(\bm{m}_2\pm\sqrt{-1}\bm{m}_3)^k\big)_0$ where $0\le k\le n$ give a basis of $n$-th order symmetric traceless tensors. 
Thus, for any nonzero tensor $U(\mathfrak{p})$, we act several $(\sqrt{1}\mathcal{L}_2+\mathcal{L}_3)$ on it to obtain $(\bm{m}_2+\sqrt{-1}\bm{m}_3)^n$, then impose several $(\sqrt{1}\mathcal{L}_2-\mathcal{L}_3)$ to obtain the whole basis.

\section{Decomposition of a tensor into symmetric traceless tensors}
Let us consider the decomposition of a general $r$-th order tensor $X$.
We start from extracting the symmetric part $X_{\mathrm{sym}}$.
The difference $X-X_{\mathrm{sym}}$ can be expressed by several terms of the form
$$
X_{\ldots i\ldots j\ldots}-X_{\ldots j\ldots i\ldots}. 
$$
For any second order tensor $Q$, its antisymmetric part is
$$
Q_{ij}-Q_{ji}=
\left(
\begin{array}{ccc}
  0 & Q_{12}-Q_{21} & Q_{13}-Q_{31}\\
  Q_{21}-Q_{12} & 0 & Q_{23}-Q_{32}\\
  Q_{31}-Q_{13} & Q_{32}-Q_{23} & 0\\
\end{array}
\right)
=\epsilon_{ijk}v_k,\quad v=\left(
\begin{array}{c}
  Q_{23}-Q_{32}\\
  Q_{31}-Q_{13}\\
  Q_{12}-Q_{21}
\end{array}
\right). 
$$
Thus, if $X$ is $r$-th order, we have the following expression, 
$$
X_{\ldots i\ldots j\ldots}-X_{\ldots j\ldots i\ldots}=\epsilon_{ijk}Z_{k\ldots}, 
$$
where $Z$ is an $(r-1)$-th order tensor.
Therefore, we arrive at 
\begin{align}
  (X-X_{\mathrm{sym}})_{j_1\ldots j_r}=\sum_{\substack{\{\tau_1,\tau_2\}\cup\{\sigma_1,\ldots\sigma_{r-2}\}\\=\{1,\ldots,r\}}}\epsilon_{j_{\tau_1}j_{\tau_2}\nu}Z_{\nu j_{\sigma_1}\ldots j_{\sigma_{r-2}}}. 
\end{align}
In the above, we use the notation $Z$ for any tensor.
Then, we can repeat this action for each $Z$, decomposing it into its symmetric part and some tensors with lower order.
We shall keep doing it until each tensor becomes symmetric.
Note that two $\epsilon_{ijk}$ can be expressed by some $\delta$: 
$$
\epsilon_{i_1j_1k_1}\epsilon_{i_2j_2k_2}=\left|
\begin{array}{ccc}
  \delta_{i_1i_2} & \delta_{i_1j_2} & \delta_{i_1k_2}\\
  \delta_{j_1i_2} & \delta_{j_1j_2} & \delta_{j_1k_2}\\
  \delta_{k_1i_2} & \delta_{k_1j_2} & \delta_{k_1k_2}
\end{array}
\right|. 
$$
So, if in any term there is no less than two $\epsilon_{ijk}$, we write them into some $\delta$. For example, 
\begin{align*}
  \epsilon_{j_1j_2\nu}\epsilon_{\nu j_3\nu'}Z_{\nu' j_4\ldots j_r}=&(\delta_{j_1j_3}\delta_{j_2\nu'}-\delta_{j_2j_3}\delta_{j_1\nu'})Z_{\nu' j_4\ldots j_r}\\
  =&\delta_{j_1j_3}Z_{j_2 j_4\ldots j_r}-\delta_{j_2j_3}Z_{j_1 j_4\ldots j_r}\\
  \epsilon_{j_1j_2\nu}\epsilon_{j_3j_4\nu'}Z_{\nu\nu' j_5\ldots j_r}=&(\delta_{j_1j_3}\delta_{j_2j_4}-\delta_{j_1j_4}\delta_{j_2j_3})Z_{\nu\nu j_5\ldots j_r}\\
  &+\delta_{j_1j_4}Z_{j_3j_2j_5\ldots j_r}+\delta_{j_2j_3}Z_{j_4j_1 j_5\ldots j_r}\\
  &-\delta_{j_1j_3}Z_{j_4j_2j_5\ldots j_r}-\delta_{j_2j_4}Z_{j_3j_1 j_5\ldots j_r}.
\end{align*}
Thus, we could write each term as the above so that there is at most one $\epsilon_{ijk}$. 
Eventually, we get the following form, 
\begin{align}
  X_{j_1\ldots j_r}=\sum_{\substack{0\le s\le r, s\text{ even}\\ \{\tau_1,\ldots,\tau_s\}\cup\{\sigma_1,\ldots\sigma_{r-s}\}\\=\{1,\ldots,r\}}}
  &\delta_{j_{\tau_1}j_{\tau_2}}\ldots\delta_{j_{\tau_{s-1}}j_{\tau_s}}Z_{j_{\sigma_1}\ldots j_{\sigma_{r-s}}}\nonumber\\
  &+\epsilon_{j_{\tau_1}j_{\tau_2}\nu}\delta_{j_{\tau_3}j_{\tau_4}}\ldots\delta_{j_{\tau_{s-1}}j_{\tau_s}}Z_{\nu j_{\sigma_1}\ldots j_{\sigma_{r-s}}}.\label{symdecomp}
\end{align}
Here, all the tensors $Z$ are symmetric tensors. 
Based on \eqref{symdecomp}, we could write each $Z$ as $U+\mathfrak{i} Z_1$ where $U$ is symmetric traceless, using Proposition \ref{U0}. We might obtain another type of term.
For example, when decomposing $\epsilon_{j_1j_2\nu}Z_{\nu j_3\ldots j_r}$, it will yield a term 
\begin{align*}
  \epsilon_{j_1j_2\nu}\delta_{\nu j_3}(Z_1)_{j_4\ldots j_r}=\epsilon_{j_1j_2j_3}(Z_1)_{j_4\ldots j_r}. 
\end{align*}
Therefore, $X$ is written as \eqref{symtrlsdecomp}. 

\section{Linearly independent terms in the expansion of $\mathscr{M}_4^{0,0,0}$}
In this section, we look into \eqref{terms_a4} and find out the linearly independent terms. 
We begin with two equalities. 
\begin{lemma}
  Suppose $Q_i$ are second order symmetric traceless tensors; $\bm{p}_i$ are vectors. 
  Then we have
  \begin{align}
    &2\mathrm{tr}(Q_1Q_2Q_3Q_4+Q_1Q_2Q_4Q_3+Q_1Q_3Q_2Q_4)\nonumber\\
    &\qquad\qquad=\mathrm{tr}(Q_1Q_2)\mathrm{tr}(Q_3Q_4)+\mathrm{tr}(Q_1Q_3)\mathrm{tr}(Q_2Q_4)+\mathrm{tr}(Q_1Q_4)\mathrm{tr}(Q_2Q_3),\\
    &(\bm{p}_1\times\bm{p}_2)\otimes\bm{p}_3+(\bm{p}_2\times\bm{p}_3)\otimes\bm{p}_1+(\bm{p}_3\times\bm{p}_1)\otimes\bm{p}_2\nonumber\\
    &+\bm{p}_3\otimes(\bm{p}_1\times\bm{p}_2)+\bm{p}_1\otimes(\bm{p}_2\times\bm{p}_3)+\bm{p}_2\otimes(\bm{p}_3\times\bm{p}_1)=\mathrm{det}(\bm{p}_1,\bm{p}_2,\bm{p}_3)\mathfrak{i}. 
  \end{align}
  Here, $Q_1Q_2$ is understood as matrix product, $\mathrm{tr}$ is the trace of a matrix, and $\times$ is the cross product of vectors in $\mathbb{R}^3$. 
\end{lemma}
\begin{proof}
  For any two symmetric traceless tensors $Q$ and $B$, we have 
  \begin{equation}
    2\mathrm{tr}(Q^3B)=\mathrm{tr}(Q^2)\mathrm{tr}(QB). 
  \end{equation}
  It can be verified by diagonalizing $Q$. 
  Then, let $Q=B=Q_1+Q_2$ to derive 
  \begin{align}
    &2\mathrm{tr}(2Q_1^2Q_2^2+Q_1Q_2Q_1Q_2)=2(\mathrm{tr}(Q_1Q_2))^2+\mathrm{tr}(Q_1^2)\mathrm{tr}(Q_2^2).
  \end{align}
  Substituting $Q_1$ with $Q_1+Q_3$, we deduce that 
  \begin{align}
    &2\mathrm{tr}(2Q_1Q_3Q_2^2+Q_1Q_2Q_3Q_2)=2\mathrm{tr}(Q_1Q_2)\mathrm{tr}(Q_2Q_3)+\mathrm{tr}(Q_1Q_3)\mathrm{tr}(Q_2^2).
  \end{align}
  Finally, substitute $Q_2$ with $Q_2+Q_4$ to obtain what is stated in the lemma.

  The second equality can be verified directly. 
\end{proof}

To simplify the notation, from now on, we omit the tensor order of $U_i^{n_i}$, i.e. write $U_i^{n_i}$ in short as $U_i$. 
Although we do not write out, we always use $n_i$ as the order of $U_i$. 
The above lemma leads to 
\begin{align}
  &2\mathpzc{a}_4(U_1,U_2,U_3,U_4;l_{12}+1,l_{13}+1,l_{14},l_{23},l_{24}+1,l_{34}+1)\nonumber\\
  &+2\mathpzc{a}_4(U_1,U_2,U_3,U_4;l_{12}+1,l_{13},l_{14}+1,l_{23}+1,l_{24},l_{34}+1)\nonumber\\
  &+2\mathpzc{a}_4(U_1,U_2,U_3,U_4;l_{12},l_{13}+1,l_{14}+1,l_{23}+1,l_{24}+1,l_{34})\nonumber\\
  =&\mathpzc{a}_4(U_1,U_2,U_3,U_4;l_{12}+2,l_{13},l_{14},l_{23},l_{24},l_{34}+2)\nonumber\\
  &+\mathpzc{a}_4(U_1,U_2,U_3,U_4;l_{12},l_{13}+2,l_{14},l_{23},l_{24}+2,l_{34})\nonumber\\
  &+\mathpzc{a}_4(U_1,U_2,U_3,U_4;l_{12},l_{13},l_{14}+2,l_{23}+2,l_{24},l_{34}).\label{a4noepsrel}
\end{align}
In the above, $n_1+n_2+n_3+n_4$ is even. 
Thus, in $\mathpzc{a}_4(U_i;l_{ij})$, for all the terms with $l_{12},l_{34}\ge 2$, they can be expressed linearly by those with $\min\{l_{12},l_{34}\}\le 1$. 
From \eqref{lrelation}, we have $n_1+n_2-2l_{12}=n_3+n_4-2l_{34}$.
So we can choose the terms where 
\begin{align}
  l_{12}\le 1\text{ if }n_1+n_2\le n_3+n_4;\ l_{34}\le 1,\text{ if }n_1+n_2\ge n_3+n_4. \label{a4noepsl}
\end{align}
Here, we notice that $l_{12}=l_{34}$ if $n_1+n_2=n_3+n_4$. 

For terms involving $\epsilon$, the lemma implies 
\begin{align}
  \mathpzc{a}_4(U_1,&U_2,U_3,U_4;l_{12},l_{13},l_{14}+1,l_{23},l_{24},l_{34},(123))\nonumber\\
  &-\mathpzc{a}_4(U_1,U_2,U_3,U_4;l_{12},l_{13}+1,l_{14},l_{23},l_{24},l_{34},(124))\nonumber\\
  &+\mathpzc{a}_4(U_1,U_2,U_3,U_4;l_{12}+1,l_{13},l_{14},l_{23},l_{24},l_{34},(134))=0,\nonumber\\
  \mathpzc{a}_4(U_1,&U_2,U_3,U_4;l_{12},l_{13},l_{14},l_{23},l_{24}+1,l_{34},(123))\nonumber\\
  &-\mathpzc{a}_4(U_1,U_2,U_3,U_4;l_{12},l_{13},l_{14},l_{23}+1,l_{24},l_{34},(124))\nonumber\\
  &-\mathpzc{a}_4(U_1,U_2,U_3,U_4;l_{12}+1,l_{13},l_{14},l_{23},l_{24},l_{34},(234))=0,\nonumber\\
  \mathpzc{a}_4(U_1,&U_2,U_3,U_4;l_{12},l_{13},l_{14},l_{23},l_{24},l_{34}+1,(123))\nonumber\\
  &+\mathpzc{a}_4(U_1,U_2,U_3,U_4;l_{12},l_{13},l_{14},l_{23}+1,l_{24},l_{34},(134))\nonumber\\
  &-\mathpzc{a}_4(U_1,U_2,U_3,U_4;l_{12},l_{13}+1,l_{14},l_{23},l_{24},l_{34},(234))=0,\nonumber\\
  \mathpzc{a}_4(U_1,&U_2,U_3,U_4;l_{12},l_{13},l_{14},l_{23},l_{24},l_{34}+1,(124))\nonumber\\
  &-\mathpzc{a}_4(U_1,U_2,U_3,U_4;l_{12},l_{13},l_{14},l_{23},l_{24}+1,l_{34},(134))\nonumber\\
  &+\mathpzc{a}_4(U_1,U_2,U_3,U_4;l_{12},l_{13},l_{14}+1,l_{23},l_{24},l_{34},(234))=0.\label{a4epsrel}
\end{align}
Notice that $n_1+n_2+n_3+n_4$ is odd in these equalities. 
Similarly, in $\mathpzc{a}_4(U_i;l_{ij},(\tau_1\tau_2\tau_3))$, we can choose the terms where 
\begin{align}
  &(\tau_1\tau_2\tau_3)=(123),(124),\ l_{34}=0,\text{ if }n_1+n_2\ge n_3+n_4+1,\nonumber\\
  &(\tau_1\tau_2\tau_3)=(134),(234),\ l_{12}=0,\text{ if }n_1+n_2+1\le n_3+n_4,\label{a4epsl}
\end{align}
because other terms can be linearly expressed by them. 

Let us consider the linearly independent terms in $\mathpzc{a}_4(U_i(\mathfrak{p}_i);l_{ij})$ and $\mathpzc{a}_4(U_i(\mathfrak{p}_i);l_{ij},(\tau_1\tau_2\tau_3))$. Here, we use the same approach as in Theorem \ref{twoforms}. 
\begin{theorem}
  Let $U_2,U_3,U_4$ be fixed and $U_1\in \mathbb{W}^{n_1}$ where $n_1$ takes all the possible values. The terms $\mathpzc{a}_4(U_i(\mathfrak{p}_i);l_{ij})$ with the condition \eqref{a4noepsl}, and $\mathpzc{a}_4(U_i(\mathfrak{p}_i);l_{ij},(\tau_1\tau_2\tau_3))$ with the condition \eqref{a4epsl}, are linearly independent. 
\end{theorem}
\begin{proof}
Recall that these terms can express \eqref{M4expansion} linearly. 
In \eqref{M4expansion}, the tensor $Y_2\otimes Y_3\otimes Y_4$ has $(2n_2+1)(2n_3+1)(2n_4+1)$ choices. 
In what follows, we show that the number of terms expressed by $U_i$ is exactly $(2n_2+1)(2n_3+1)(2n_4+1)$. 

We use induction on $n_2$ and $n_3$. 
When any of $n_2$, $n_3$ or $n_4$ is zero, it reduces to the case $\mathpzc{a}_3$ (see the discussion below \eqref{terms_a3p}).
So, we discuss the case where $n_2,n_3,n_4\ge 1$. 
If $l_{23}\ge 1$, the number of terms equals to the case where $U_2$, $U_3$ and $U_4$ are of the order $n_2-1$, $n_3-1$, and $n_4$, respectively, which is $(2n_2-1)(2n_3-1)(2n_4+1)$ by the assumption of induction. 
Now let $l_{23}=0$. To count the number, we use \eqref{lrelation} and notice the constraints $l_{ij}\ge 0$. There are six cases: 
\begin{enumerate}
\item $\mathpzc{a}_4(U_i(\mathfrak{p}_i);l_{ij})$ where $l_{12}=0$ or $l_{34}=0$. In this case, $n_1+n_2+n_3+n_4$ is even, and 
  $$
  l_{12}+l_{24}=n_2,\ l_{13}+l_{34}=n_3. 
  $$
  \begin{enumerate}[(a)]
  \item When $l_{12}=0$, we solve $l_{24}=n_2$, and 
  $$
  l_{34}=\frac{n_3+n_4-n_1-n_2}{2},\ l_{13}=\frac{n_1+n_2+n_3-n_4}{2},\ l_{14}=\frac{n_1+n_4-n_2-n_3}{2}. 
  $$
  It yields
  $$
  n_1\le n_3+n_4-n_2,\ n_1\ge n_2+n_3-n_4,\ n_1\ge n_4-n_3-n_2. 
  $$
  \item When $l_{34}=0$, we solve $l_{13}=n_3$, and 
  $$
  l_{12}=\frac{n_1+n_2-n_3-n_4}{2},\ l_{24}=\frac{n_2+n_3+n_4-n_1}{2},\ l_{14}=\frac{n_1+n_4-n_2-n_3}{2}. 
  $$
  It yields
  $$
  n_1\ge n_3+n_4-n_2,\ n_1\le n_2+n_3+n_4,\ n_1\ge n_2+n_3-n_4. 
  $$
  \end{enumerate}
  We combine (a) and (b). If $n_3+n_4-n_2<0$, then the range of $n_1$ is $n_2+n_3-n_4\le n_1\le n_2+n_3+n_4$. If $n_3+n_4-n_2\ge 0$, then $|n_2+n_3-n_4|\le n_1\le n_2+n_3+n_4$. So, we have $|n_2+n_3-n_4|\le n_1\le n_2+n_3+n_4$ where $n_1$ has the same parity as $n_2+n_3+n_4$. 
\item $\mathpzc{a}_4(U_i(\mathfrak{p}_i);l_{ij})$ where $\min\{l_{12},l_{34}\}= 1$. 
  Similar to the above, we deduce that $1+|n_2+n_3-n_4-1|\le n_1\le n_2+n_3+n_4-2$ where $n_1$ has the same parity as $n_2+n_3+n_4$. 
\item $\mathpzc{a}_4(U_i(\mathfrak{p}_i);l_{ij},(123))$ where $l_{34}=0$.
  We solve that $l_{13}=n_3-1$, and 
  $$
  l_{12}=\frac{n_1+n_2-n_3-n_4-1}{2},\ l_{24}=\frac{n_2+n_3+n_4-n_1-1}{2},\ l_{14}=\frac{n_1+n_4-n_2-n_3+1}{2}. 
  $$
  It yields
  $$
  n_1\ge n_3+n_4-n_2+1,\ n_1\le n_2+n_3+n_4-1,\ n_1\ge n_2+n_3-n_4-1. 
  $$
\item $\mathpzc{a}_4(U_i(\mathfrak{p}_i);l_{ij},(124))$ where $l_{34}=0$.
  We solve that $l_{13}=n_3$, and 
  $$
  l_{12}=\frac{n_1+n_2-n_3-n_4-1}{2},\ l_{24}=\frac{n_2+n_3+n_4-n_1-1}{2},\ l_{14}=\frac{n_1+n_4-n_2-n_3-1}{2}. 
  $$
  It yields 
  $$
  n_1\ge n_3+n_4-n_2+1,\ n_1\le n_2+n_3+n_4-1,\ n_1\ge n_2+n_3-n_4+1. 
  $$
\item $\mathpzc{a}_4(U_i(\mathfrak{p}_i);l_{ij},(134))$ where $l_{12}=0$.
  We solve that $l_{24}=n_2$, and 
  $$
  l_{34}=\frac{n_3+n_4-n_1-n_2-1}{2},\ l_{24}=\frac{n_1+n_2+n_3-n_4-1}{2},\ l_{14}=\frac{n_1+n_4-n_2-n_3-1}{2}. 
  $$
  It yields
  $$
  n_1\le n_3+n_4-n_2-1,\ n_1\ge n_4-n_2-n_3+1,\ n_1\ge n_2+n_3-n_4+1. 
  $$
\item $\mathpzc{a}_4(U_i(\mathfrak{p}_i);l_{ij},(234))$ where $l_{12}=0$.
  We solve that $l_{24}=n_2-1$, and 
  $$
  l_{34}=\frac{n_3+n_4-n_1-n_2-1}{2},\ l_{24}=\frac{n_1+n_2+n_3-n_4-1}{2},\ l_{14}=\frac{n_1+n_4-n_2-n_3-1}{2}. 
  $$
  It yields
  $$
  n_1\le n_3+n_4-n_2-1,\ n_1\ge n_4-n_2-n_3+1,\ n_1\ge n_2+n_3-n_4-1. 
  $$
\end{enumerate}
In cases 3 to 6, $n_1$ has the different parity from $n_2+n_3+n_4$. 
Combine case 3 and case 6. If $n_2>n_4$, then case 6 is empty, and we have $n_2+n_3-n_4-1\le n_1\le n_2+n_3+n_4-1$. If $n_2\le n_4$, then case 6 is $|n_2+n_3-n_4-1|\le n_1\le n_3+n_4-n_2-1$, and case 3 is $n_3+n_4-n_2+1\le n_1\le n_2+n_3+n_4-1$.
So, we arrive at $|n_2+n_3-n_4-1|\le n_1\le n_2+n_3+n_4-1$. 
Similarly, we combine case 4 and case 5 to have $|n_2+n_3-n_4|+1\le n_1\le n_2+n_3+n_4-1$. 

So, let us combine cases 1, 4 and 5. The range of $n_1$ is $|n_2+n_3-n_4|\le n_1\le n_2+n_3+n_4$. 
The cases 2, 3 and 6 lead to $|n_2+n_3-n_4-1|\le n_1\le n_2+n_3+n_4-1$. 
Thus, the number of terms is
\begin{align*}
  &\sum_{r=|n_2+n_3-n_4|}^{n_2+n_3+n_4}(2r+1)+\sum_{r=|n_2+n_3-n_4-1|}^{n_2+n_3+n_4-1}(2r+1)\\
  =&4(n_2+n_3)(2n_4+1)\\
  =&(2n_2+1)(2n_3+1)(2n_4+1)-(2n_2-1)(2n_3-1)(2n_4+1). 
\end{align*}
This concludes the proof. 
\end{proof}

The theorem indicates that \eqref{a4noepsrel} and \eqref{a4epsrel} give all the linear relations without missing anything. 
Now, we consider \eqref{terms_a4} where label permutations are taken into consideration. 
As we have discussed in the main text, once the tensors appearing in $\mathpzc{a}_4$ are chosen, we can arrange them in the order we want. 
When $U_i$ are mutually unequal, the conditions \eqref{a4noepsl} and \eqref{a4epsl} are just those in the Table \ref{sum_term}. 
However, these conditions are not suitable if some of $U_i$ are equal. 

We omit the case where $U_1=U_2\ne U_3,U_4$, and only examine the case $U_1=U_2=U_3$. 

\textbf{Problem 1:} Consider \eqref{terms_a4_1} where $U_1=U_2=U_3$. It is equivalent to consider the linearly independent terms of $\mathpzc{a}_4(U_1,U_1,U_1,U_4;l_{ij})$. 
The relation \eqref{lrelation} between $l_{ij}$ can be rewritten as 
\begin{align*}
  &2(l_{12}-l_{34})=2(l_{13}-l_{24})=2(l_{23}-l_{14})=n_1-n_4, \\
  &l_{14}+l_{24}+l_{34}=n_4. 
\end{align*}
Thus, we define $(d_1,d_2,d_3)=(l_{12},l_{13},l_{23})$ if $n_1\le n_4$, and $(d_1,d_2,d_3)=(l_{34},l_{24},l_{14})$ if $n_1\ge n_4$.
We have $d_1+d_2+d_3=\min\{(3n_1-n_4)/2,n_4\}\triangleq d$ is a constant determined by $n_1$ and $n_4$. 
Define $\psi(d_1,d_2,d_3)=\mathpzc{a}_4(U_1,U_1,U_1,U_4;l_{ij})$. 
Similar to \eqref{a3swlb}, we have $\psi(d_1,d_2,d_3)=\psi(d_{\sigma(1)},d_{\sigma(2)},d_{\sigma(3)})$ for any permutation $\sigma$. 
The linear relation \eqref{a4noepsrel} is then written as 
\begin{align*}
  &\psi(d_1+2,d_2,d_3)+\psi(d_1,d_2+2,d_3)+\psi(d_1,d_2,d_3+2)\\
  =&2\psi(d_1+1,d_2+1,d_3)+2\psi(d_1+1,d_2,d_3+1)+2\psi(d_1,d_2+1,d_3+1). 
\end{align*}

According to the conditions in Table \ref{sum_term}, we need to show that $\psi(i,i,d-2i)$ for $3i\le d$ are linearly independent and can linearly express others.
We use induction on $d$. 
For $d=0,1,2$ we verify directly. 
When $d=0$, there is only one term $\psi(0,0,0)$. 
When $d=1$, by the permutational symmetry there is only one term $\psi(0,0,1)$. 
When $d=2$, we have 
\begin{align*}
  3\psi(0,0,2)=&\psi(2,0,0)+\psi(0,2,0)+\psi(0,0,2)\\
  =&2\psi(1,1,0)+2\psi(1,0,1)+2\psi(0,1,1)=6\psi(1,1,0). 
\end{align*}
Thus, there is only one linearly independent term $\psi(0,0,2)$. 

Assume $d\ge 3$. The linear relations between $\psi$ where $d_i\ge 1$ are identical to $\psi(d_1-1,d_2-1,d_3-1)$ for $d-3$. By the assumption of induction, $\psi(i,i,d-2i)$ for $1\le i\le d-2i$ give the linearly independent terms. 
Thus, let us assume that $\psi(d_1,d_2,d_3)$ are all known when $d_i\ge 1$, and solve $\psi$ when some $d_i$ are zero. 
Now let $d_1=0$. 
If $d_2,d_3\ge 1$, then 
\begin{align*}
  &\psi(0,d_2+2,d_3)-2\psi(0,d_2+1,d_3+1)+\psi(0,d_2,d_3+2)\\
  =&2\psi(1,d_2+1,d_3)+2\psi(1,d_2,d_3+1)-\psi(2,d_2,d_3) 
\end{align*}
is known. 
For $d_2=0$, we have
\begin{align*}
  \psi(0,2,d-2)-2\psi(0,1,d-1)+\psi(0,0,d)=2\psi(1,1,d-2)+2\psi(1,0,d-1)-\psi(2,0,d-2). 
\end{align*}
Use invariance under permutation, we get 
\begin{align*}
  -\psi(0,2,d-2)+2\psi(0,1,d-1)=-\psi(1,1,d-2)+\frac{1}{2}\psi(0,0,d). 
\end{align*}
Define a vector $\bm{z}$ where $z_i=\psi(0,i,d-i)$ for $i=1,\ldots, d-1$.
The above linear equations can be written as 
\begin{align*}
  \left(
  \begin{array}{cccc}
    2 & -1 & &\\
    -1 & 2 & \ddots &\\
    & \ddots &\ddots & -1\\
    & & -1 & 2
  \end{array}
  \right)\bm{z}=\bm{b}+\left(
  \begin{array}{c}
    \frac{1}{2}\psi(0,0,d)\\
    0\\
    \vdots\\
    0\\
    \frac{1}{2}\psi(0,0,d)
  \end{array}
  \right),
\end{align*}
where $\bm{b}$ satisfies $b_i=b_{d-i}$ that is given by $\psi(d_1,d_2,d_3)$ with $d_i\ge 1$. 
Hence, the value of $\psi(0,0,d)$ is needed to fully determine $\psi(0,d_2,d_3)$, and the solution also satisfies $z_i=z_{d-i}$.

\textbf{Problem 2:} Consider \eqref{terms_a4_2} where $U_1=U_2=U_3$. Again, we shall consider the linearly independent terms of $\mathpzc{a}_4(U_1,U_1,U_1,U_4;l_{ij},(\tau_1\tau_2\tau_3))$. 
Using arguments similar to \eqref{a3swlb}, we can deduce that 
\begin{align*}
  \mathpzc{a}_4(U_1,U_1,&U_1,U_4;l_{12},l_{13},l_{14},l_{23},l_{24},l_{34},(124))\nonumber\\
  =&\mathpzc{a}_4(U_1,U_1,U_1,U_4;l_{12},l_{23},l_{24},l_{13},l_{14},l_{34},(214))\nonumber\\
  =&-\mathpzc{a}_4(U_1,U_1,U_1,U_4;l_{12},l_{23},l_{24},l_{13},l_{14},l_{34},(124)),\nonumber\\
  \mathpzc{a}_4(U_1,U_1,&U_1,U_4;l_{12},l_{13},l_{14},l_{23},l_{24},l_{34},(134))\nonumber\\
  =&\mathpzc{a}_4(U_1,U_1,U_1,U_4;l_{13},l_{12},l_{14},l_{23},l_{34},l_{24},(124)),\nonumber\\
  \mathpzc{a}_4(U_1,U_1,&U_1,U_4;l_{12},l_{13},l_{14},l_{23},l_{24},l_{34},(234))\nonumber\\
  =&-\mathpzc{a}_4(U_1,U_1,U_1,U_4;l_{23},l_{13},l_{34},l_{12},l_{24},l_{14},(124)).\nonumber
\end{align*}
Thus, it allows us not to consider the terms with $(\tau_1\tau_2\tau_3)=(134),(234)$. 
When $(\tau_1\tau_2\tau_3)=(124)$, the relations between $l_{ij}$ require
\begin{align*}
  &2(l_{12}-l_{34}+1)=2(l_{13}-l_{24})=2(l_{23}-l_{14})=n_1-n_4+1, \\
  &l_{14}+l_{24}+l_{34}=n_4-1. 
\end{align*}
Thus, we define $(d_1,d_2,d_3)=(l_{12},l_{13},l_{23})$ if $n_1\le n_4-1$, and $(d_1,d_2,d_3)=(l_{34},l_{24},l_{14})$ if $n_1\ge n_4+1$.
We have $d_1+d_2+d_3=\min\{n_4-1,(3n_1-n_4-1)/2\}=d$.
To simplify the presentation, we only discuss the case $n_1\le n_4-1$. 
Define $\varphi(d_1,d_2,d_3)=\mathpzc{a}_4(U_i^{n_i};l_{ij},(124))$. 
Then we have 
\begin{align}
  \varphi(d_1,d_2,d_3)=-\varphi(d_1,d_3,d_2). \label{oppvar}
\end{align}
Use permutational symmetry on $\mathpzc{a}_4(U_1,U_1,U_1,U_4;l_{ij},(123))$, the first three equations in \eqref{a4epsrel} become
\begin{align}
  &\varphi(d_1,d_2,d_3+1)-\varphi(d_3,d_2,d_1+1)\nonumber\\
  =&-\varphi(d_1,d_3,d_2+1)+\varphi(d_2,d_3,d_1+1)\nonumber\\
  =&\varphi(d_3,d_1,d_2+1)-\varphi(d_2,d_1,d_3+1)\nonumber\\
  =&\mathpzc{a}_4(U_1,U_1,U_1,U_4;d_1,d_2,d_3+c,d_3,d_2+c,d_1+c,(123)),\quad c=\frac{n_4+1-n_1}{2}. \label{recvar}
\end{align}
The fourth becomes 
\begin{align}
  \varphi(d_1,d_2,d_3)+\varphi(d_2,d_3,d_1)+\varphi(d_3,d_1,d_2)=0. \label{rotvar}
\end{align}

Our goal is to verify that $\varphi(i,i,d-2i)$ for $3i<d$ give all the linearly indepedent terms.
Use induction on $d$. 
When $d=0$, we have $\varphi(0,0,0)=0$. 
When $d=1$, we have $\varphi(1,0,0)=0$ and $\varphi(0,1,0)=-\varphi(0,0,1)$.
So, there is only one linearly independent term $\varphi(0,0,1)$. 
When $d=2$, we have $\varphi(2,0,0)=\varphi(0,1,1)=0$, and 
\begin{align*}
  \varphi(0,0,2)-\varphi(1,0,1)=-\varphi(0,1,1)+\varphi(0,1,1)=\varphi(1,0,1)-\varphi(0,0,2). 
\end{align*}
Together with $\varphi(0,2,0)=-\varphi(0,0,2)$, $\varphi(1,1,0)=-\varphi(1,0,1)$, we find that there is only one linearly independent term $\varphi(0,0,2)$. 

Now consider $d\ge 3$.
The linear relations between $\varphi$ for $d_i\ge 1$ are identical to the case $d-3$.
By the assumption of induction, in these terms the linearly independent ones can be given by $\varphi(i,i,d-2i)$ with $1\le i<d-2i$. 
We assume that $\varphi(d_1,d_2,d_3)$ are known for $d_i\ge 1$ and solve those with some $d_i=0$. 
If two of $d_i$ are zero, the linear relations yield 
\begin{align*}
  \varphi(d,0,0)=0,
  \varphi(0,0,d)=-\varphi(0,d,0).
\end{align*}
Below, we consider $\psi$ with exactly one $d_i$ zero, to show that they can be solved from $\varphi(0,0,d)$ and $\varphi(d_1,d_2,d_3)$ where $d_i\ge 1$. 

In \eqref{recvar}, let $d_3=0$, $d_1+d_2=d-1$, where $1\le d_1\le d_2\le d-2$.
Then, the first and third lines give 
\begin{align*}
  \varphi(0,d_1,d_2+1)+\varphi(0,d_2,d_1+1)=\varphi(d_1,d_2,1)+\varphi(d_2,d_1,1), 
\end{align*}
where the right-hand side is known. 
Together with $\varphi(0,d_1,d_2)=-\varphi(0,d_2,d_1)$, 
we can solve $\varphi(0,d_1,d_2)$ for $1\le d_1,d_2\le d-1$. 

Next, we deal with $\varphi(d_1,0,d_2)$ where $d_1,d_2\ge 1$.
Using the second line in \eqref{recvar}, we obtain 
\begin{align}
  \varphi(d_2,0,d_1+1)-\varphi(d_1,0,d_2+1)=\varphi(d_1,d_2,1)-\varphi(0,d_2,d_1+1), \label{e1}
\end{align}
where the right-hand side is alrealy obtained above. 
Note that switching $d_1$ and $d_2$ leads to the same equation, and $d_1=d_2$ gives nothing. 
So, we require $0\le d_1<d_2\le d-1$. 
Here, $d_1=0$ gives $\varphi(0,0,d)=\varphi(d-1,0,1)$. 
Then, by \eqref{oppvar} and \eqref{rotvar}, we deduce that 
\begin{align}
  \varphi(d_2,0,d_1)=-\varphi(d_1,d_2,0)-\varphi(0,d_1,d_2)
  =\varphi(d_1,0,d_2)-\varphi(0,d_1,d_2), \label{e2}
\end{align}
where $1\le d_1< d_2\le d-1$. 
\eqref{e1} and \eqref{e2} give $d-1$ equations in total for $\varphi(d_1,0,d_2)$ where $d_1,d_2\ge 1$. 
They can indeed be solved by rewriting the left-hand side of \eqref{e1} as 
\begin{align*}
  \varphi(d_2,0,d_1+1)-\varphi(d_1,0,d_2+1)=\varphi(d_1+1,0,d_2)-\varphi(d_1,0,d_2+1)-\varphi(0,d_1+1,d_2),
\end{align*}
leading to 
\begin{align*}
  \varphi(d_1+1,0,d_2)-\varphi(d_1,0,d_2+1)=\varphi(d_1,d_2,1). 
\end{align*}

Finally, we use $\varphi(d_1,d_2,0)=-\varphi(d_1,0,d_2)$ for $1\le d_1,d_2\le d-1$ to finish the induction. 

\bibliographystyle{plain}
\bibliography{bib_sym}

\end{document}